\documentclass{easychair}
\usepackage{doc}

\usepackage{graphicx}
\usepackage{underscore}       
\usepackage{tablefootnote}
\usepackage{algorithmic}
\usepackage[utf8]{inputenc}    % utf8 support       
\usepackage[T1]{fontenc}       % code for pdf file  
\usepackage{amssymb}
\usepackage{amsthm}
\usepackage{multirow}
\usepackage{pgf}
\usepackage{tikz}
\usetikzlibrary{arrows,automata,positioning}
\usepackage{stmaryrd}
\usepackage{mathtools}
\usepackage{appendix,xspace}
\usepackage{multirow,hyperref}
\usepackage{mathpartir}

\newcommand{\dui}
{\hyperref[dui]{\textbf{dui}}\xspace}

\newtheorem{definition}{Definition}
\newtheorem{remark}{Remark}
\newtheorem{example}{Example}
\newtheorem{theorem}{Theorem}
\newtheorem{lemma}{Lemma}

\usepackage{hyperref}
\hypersetup{
	colorlinks,
	citecolor=brown,
	final,
	linkcolor=blue,
}
\newcommand\calF{{\cal F}}
\newcommand\calG{{\cal G}}
\newcommand\calH{{\cal H}}
\newcommand\sloc[1]{l^{\mathrm s}_{#1}}
\newcommand\dloc[1]{l^{\mathrm d}_{#1}}
\newcommand\trans{\mathbb{T}}
\usepackage[normalem]{ulem}
\usepackage{thmtools,thm-restate}
\usepackage{enumerate}
\usepackage{wrapfig}

\def\orcidID#1{\href{http://orcid.org/#1}{\raisebox{-1.25pt}{\includegraphics{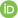}}}}

\renewcommand\arraystretch{1}
\newcommand\nstep[2]{#1^{[#2]}}
\begin{document}
	%
	%\title{Automatic Synthesis of Polynomial Probabilistic Invariants}
	\title{Polynomial Invariants for Probabilistic Transition Systems with Unbounded Support}

    \author{Anne Schreuder\inst{1}\orcidID{0009-0007-3651-9401}\and
Lorenz Winkler\inst{2}\orcidID{0009-0000-0177-2764} \and
Laura Kov\'acs\inst{2}\orcidID{0000-0002-8299-2714} \and
C.-H. Luke Ong\inst{3}\orcidID{0000-0001-7509-680X}}
	\institute{
    Aalto University, Finland\and
    TU Wien, Austria\and
    Nanyang Technological University, Singapore
}
	\titlerunning{Polynomial Invariants for PTS with Unbounded Support}
    \authorrunning{Schreuder, Winkler, Kov\'acs, Ong}
	
	\maketitle              % typeset the header of the contribution
\begin{abstract}
	We study the synthesis of polynomial invariants for probabilistic transition systems (PTS) based on martingale theory. 
	We present tractable methods to verify that such polynomials are indeed invariants, in the sense that their expected value upon termination is the same as their value at the start of the computation. 
	We do this by applying the Optional Stopping Theorem {(OST)} in the form of a specific precondition.
    This precondition requires the existence of an integrable \textbf{d}ominating function for the martingale expression, which implies \textbf{u}niform  \textbf{i}ntegrability; we refer to this condition as {\bf dui}.   
    For linear PTS we simplify the {\bf dui} property to proving finiteness of the expected value of an expression depending on the update matrix, the degree of the martingale expression, and the stopping time. 
    {Specifically, if all random samples have finite moments and we can verify a moment bound on the runtime of a linear loop, then we can automatically synthesise polynomial loop invariants that satisfy the OST.}
    Notably, {\bf dui}
    allows for the sampled distributions to have unbounded support, {which is a novel contribution to the field.}  
    %{\color{blue} we only require these sampled distributions to have finite moments.}
    %it only requires the {\color{blue} sampled distributions, i.e.~the distributions from which random samples are taken,} to have finite moments. 
%
    We  implemented and evaluated our method
    to symbolically compute
    upper and lower bounds for moments of program variables upon termination. 
\end{abstract}
	%
	%
	%
	%%%%%%%%%%%%%%%%%%%%%%%%%%%%%%%%%%%%%%%%%%%%%%%%%%%%%%%%%%%%%%%%%%%%%%%%%%%%%
	% Introduction %
	%%%%%%%%%%%%%%%%%%%%%%%%%%%%%%%%%%%%%%%%%%%%%%%%%%%%%%%%%%%%%%%%%%%%%%%%%%%%%%
	
\section{Introduction}
	
Probabilistic programs are computer programs that make random choices. 
Their increasingly widespread use in many areas of computer science --- ranging across machine learning, network protocols and robotics --- has led to a recent surge of interest in the semantics of probabilistic programs and methods for reasoning about them.
% In this paper we present useful methods and an algorithmic solution to a central verification problem: \emph{automatic extraction of program invariants}.
In this paper we present a theoretical result for invariants of probabilistic programs that sample from distributions with \emph{unbounded support}. We then show the applicability of our result by \emph{synthesising polynomial invariants} for a restricted class of loops and use them to \emph{automatically derive symbolic bounds} for the expected value of monomials after the loop has terminated.

\paragraph{Problem setting.} 
Following \cite{forsyte,g,DBLP:conf/atva/FengZJZX17,DBLP:conf/cav/ChenHWZ15,DBLP:conf/qest/GretzKM13,DBLP:conf/sas/KatoenMMM10,b}, we explore methods that automatically synthesise expressions over the program variables whose expectation is invariant over the course of the computation. 
%Similar work has been done in \cite{forsyte,g,DBLP:conf/atva/FengZJZX17,DBLP:conf/cav/ChenHWZ15,DBLP:conf/qest/GretzKM13,DBLP:conf/sas/KatoenMMM10,b}.
Mathematically, the invariants we construct are \emph{martingales}. %which are sequences of random variables that are expectation invariant over the computation steps.
Elegant and effective, martingales are commonplace in the static analysis of probabilistic programs %and have been widely studied 
\cite{a,c,d,DBLP:journals/pacmpl/HarkKGK20,DBLP:conf/tacas/KuraUH19,g,h,DBLP:journals/pacmpl/Huang0CG19,DBLP:journals/corr/abs-2010-06367,DBLP:journals/pacmpl/BatzKKM21,DBLP:journals/corr/abs-2010-03444,DBLP:conf/vmcai/FuC19,DBLP:conf/popl/ChatterjeeNZ17}.

We restrict ourselves to invariants of probabilistic programs that have the same expected value before and after the program has run.
%	For non-probabilistic programs this distinction does not exist. 
%	However, probabilistic programs may stop after different number of steps depending on the random samples made. 
In general, the runtime of a probabilistic program is not a concrete number but a probability distribution.
Consequently, the expectation of a step-wise invariant at the point of termination (\emph{qua} stopping time) may differ from its expectation after any fixed number of steps (see Example~\ref{example: gambling}). 
For non-probabilistic programs this effect does not exist.
First observed in betting strategies, this phenomenon has been analysed extensively in probability theory, leading to Doob's celebrated Optional Stopping Theorem (OST)~e.g.~\cite[Theorem 5.3.1]{cohen}. 
This theorem yields preconditions for the expectation of a martingale %random variable
taken at a stopping time to coincide with the expectation taken at the beginning.
%after the initial expression.
%

The Optional Stopping Theorem (OST) has recently been successfully applied to the static analysis of probabilistic programs~\cite{costanalysis,wang2021,DBLP:journals/pacmpl/BatzKKM21,DBLP:conf/atva/FengZJZX17,piecewise_analysis_prob2026,chatterjee2024quantitative}. 
However the phenomenon of the necessity of an ``extra step'', that is verifying a precondition of the OST, 
%required 
%to check suitable preconditions 
had already been observed in the analysis of nested loops~\cite{wp1,DBLP:journals/pacmpl/HarkKGK20}.
When reasoning about such probabilistic programs, it is common to use inductive arguments for the inner loop. 
For reasons similar to those described earlier, it is necessary to impose appropriate OST-type preconditions on the inner-loop invariants, in order to formulate sound proof rules. 
These methods have successfully been applied in
\cite{wp1,DBLP:conf/atva/FengZJZX17,DBLP:conf/cav/ChenHWZ15,DBLP:conf/qest/GretzKM13,DBLP:conf/vmcai/FuC19}. 
In particular, for nested loop analysis it is important to handle random updates with unbounded support. Otherwise, only loops that are entered at most $N$ times can be analysed. Although our work  mainly restricts to single loop programs, our approach does not require bounds on the support of the sampled distributions. 
%As  argued by~\cite{DBLP:journals/pacmpl/Huang0CG19}, great care is needed when reasoning about probabilistic programs.

\begin{wrapfigure}{r}{0.45\textwidth}
\vspace{-10pt}
    \begin{algorithmic}
        \STATE {$x := x_0$}
        \STATE{$y := y_0$}
        \STATE{$z := z_0$}
        \STATE{$k:=0$}
        \WHILE {$x \geq 0$} 
        \STATE{$k:= k+1$}
        \STATE{$u \sim \mathrm{Uniform}(-1,0)$}
        \STATE {$x := x + u$}
        \STATE {$n_1\sim\mathrm{Normal}(1,2)$}
        \STATE {$y:=y+n_1+u$}
        \STATE {$n_2\sim\mathrm{Normal}(-2,4)$}
        \STATE{$z:=z+y+n_2$}
        \ENDWHILE	
    \end{algorithmic}
    \caption{Linear PTS  with updates over distributions with unbounded support.}
    \label{fig:running-example-linear-loop}
\end{wrapfigure}

\paragraph{Our approach.}
We work with \emph{probabilistic transition systems} (PTS) \cite{a,b,h,DBLP:conf/cav/ChenHWZ15}, a general model of probabilistic imperative programs. 
Any standard imperative program with random samples and (nested) loops can be expressed as a PTS. An example of such a program is given in Figure~\ref{fig:running-example-linear-loop}.
In particular, we allow our programs to manipulate values in the field $\mathbb{R}$ of real numbers and to sample from continuous and unbounded probability distributions. Figure~\ref{fig:running-example-linear-loop}, for instance, samples from a continuous uniform and a normal distribution.

We first present a \emph{precondition of the OST (Theorem~\ref{thm: ost}), which we call {\dui}}. The {\dui} precondition requires the existence of an integrable \textbf{d}ominating function (for the martingale expressions) which implies the martingale is \textbf{u}niformly {\bf i}ntegrable.  
%To the best of our knowledge, our {\bf dui} precondition is novel. 
A key advantage of {\bf dui} is {its generality and its flexibility. 
It does not impose a constant bound on the stopping time (runtime), the program variables, or their updates.}
%
%is that it does not impose any restriction on the stopping time (runtime), nor does it require the program variables to be bounded by a constant. 
As such, our work differs from approaches requiring a constant bound on an invariant    
%using OST preconditions that do not involve a stopping time via uniform integrability \cite{c}, or bounded by a constant 
\cite{DBLP:journals/pacmpl/HarkKGK20,DBLP:conf/atva/FengZJZX17,DBLP:conf/cav/ChenHWZ15,DBLP:conf/qest/GretzKM13}, which is a special case of the {\dui} precondition, or bounded updates \cite{wang2021,costanalysis,piecewise_analysis_prob2026}.

While the {\dui} precondition  is rather abstract, {it is more concrete than the condition of uniform integrability which it implies.} 
{Hence, it can} be used to derive {more tractable} sufficient conditions for invariants satisfying the OST.
%for certain classes of loops.
We demonstrate this by deriving a sufficient condition for polynomial invariants (expressible as a martingale) of simple linear programs {consisting of a single loop} to satisfy the OST (Theorem~\ref{thm:moment-monomials}). 
This {novel} condition is twofold: 
First, the distributions over which we sample must be identically distributed in every iteration, and all moments of those distributions must exist. 
{This is a strict improvement on the state-of-the-art which requires all sample distributions to be bounded \cite{wang2021,costanalysis,piecewise_analysis_prob2026}.}
Second, the runtime must satisfy an expectation bound that depends on the degree of the desired polynomial invariant, the structure of the update matrix, and whether variables from distributions with unbounded support affect the runtime.
We prove this %Theorem~\ref{thm: moment bound}
theorem with the {\dui} precondition.
Moreover, our method can deduce upper and lower bounds on the expected values of higher moments of the program variables {and of the runtime\footnote{
    It is always possible to introduce a new variable that updates by one each time the loop is entered --- effectively storing the runtime.
    This is the role of the variable $k$ in Figure~\ref{fig:running-example-linear-loop}.
} at the end of the computation --- potentially improving the initial bound on the runtime}; bounds upon the moments of Figure~\ref{fig:running-example-linear-loop} are listed in Table~\ref{tab:experimental-results}.

\paragraph{Our contributions.}
%We work with \emph{probabilistic transition systems} (PTS) \cite{a,b,h,DBLP:conf/cav/ChenHWZ15}, a general model of probabilistic imperative programs. 
%
We bring the following main contributions\footnote{with detailed proofs in the appendix}. 
\begin{itemize}
    \item We focus on a {very flexible} precondition for the OST, which requires the existence of an integrable \textbf{d}ominating function for the martingale expression implying that the martingale is \textbf{u}niformly  \textbf{i}ntegrable -- we refer to this condition as {\dui} (Section~\ref{sec: background}). To the best of our knowledge, our {\dui} condition is new {in the study of probabilistic programs} and has the advantage that it does not impose any restrictions on the stopping time {a priori\footnote{
        This precondition is more abstract. Verifying it can sometimes require bounds on the runtime as seen in Section~\ref{subsection:constructing-martingales-linear-programs}. However, there may be other applications of this precondition where such a step is not necessary.
    }}.
    \item We present a sufficient condition for the application of the OST for polynomial martingales of simple linear loops (Section~\ref{sec: linear}). Checking the precondition then simplifies to proving finiteness of moments of the stopping time {(Theorem~\ref{thm:moment-monomials})}.
    \item Taking the number of finite moments of the stopping time as an input, we construct martingale expressions, which satisfy the OST based on the {\dui} condition (Section~\ref{subsection:constructing-martingales-linear-programs}). Those martingales can then be used to symbolically compute upper and lower bounds for the moments of program variables upon termination. {Notably this includes variables having negative, unbounded increments that depend on the program state.} 
    %\item {}
    \item While we make no claims about completeness of our method and still require (potentially manual) proofs of finiteness of stopping times in advance, we give examples where our method is applicable (Section~\ref{section:experimental-results}). We implemented our work on top of the \texttt{Polar} analyser~\cite{DBLP:moosbrugger-polar} and evaluated it on programs with distributions of unbounded support (Table~\ref{tab:experimental-results}). {To the best of our knowledge, existing approaches for cost analysis or bound computation in probabilistic programs cannot handle our examples with unbounded support.} {Our implementation and instructions for reproducing our results are publicly available\footnote{\url{https://github.com/probing-lab/polar/tree/166b275}}.}
\end{itemize}

	%%%%%%%%%%%%%%%%%%%%%%%%%%%%%%%%%%%%%%%%%%%%%%%%%%%%%%%%%%%%%%%%%%%%%%%%%%%%%
	% Syntax of Probabilistic Programs %
	%%%%%%%%%%%%%%%%%%%%%%%%%%%%%%%%%%%%%%%%%%%%%%%%%%%%%%%%%%%%%%%%%%%%%%%%%%%%%%
	
	\section{Probabilistic Transition Systems and Invariants} 
    \label{sec: program}\label{sec:invariant}

    As models of imperative programs we consider probabilistic transition systems (PTS).  
	We distinguish (deterministic) program variables 
	%$\boldsymbol{X} = \{x_1, x_2, \cdots, x_n\}$ 
	${X} = \{x_1, x_2, \cdots, x_n\}$
	which are assigned deterministically, and random sample variables ${R} = \{r_1, r_2, \cdots, r_m\}$ which are sampled from probability distributions. 
	All variables  take values in $\mathbb{R}$.
	We assume that the random samples are
	stochastically independent of each other and identically, independently distributed over time;
	{moreover we assume that the sampled distributions are known and that all of their moments are finite.}
	%moreover their distribution are i.i.d.~over time, and all moments exist.

	%Recall that a function $X: \Omega \rightarrow \mathbb{R}$ defined on a probability space $(\Omega, \calF, \mathbb{P})$ is a {\em random variable} if $X$ is $\mathbb{B}(\mathbb{R})$-$\calF$-measurable\footnote{The $\sigma$-algebra $\mathbb{B}(\mathbb{R})$ is the one generated by open intervals of $\mathbb{R}$.}, i.e.~for all $A \in \mathbb{B}(\mathbb{R})$ -- the $\sigma$-algebra generated by the open subsets of $X$, we have $X^{-1}(A) \in \calF$.	
	%The $k$-$th$ \emph{moment} of $X$ is the integral $\int_{\Omega} X^k \ d\mathbb{P}$, which exists if $\int_{\Omega} |X|^k \ d\mathbb{P}$ is finite.

	\begin{definition}[Probabilistic Transition Systems - PTS~\cite{a,b,h}]
		\label{def: pts}
		%%~\\
		A \emph{probabilistic transition system (PTS)} \index{transition system} \index{transition system!probabilistic}
		%$\Pi$
		is a tuple
		$\Pi = \langle W, X, R, \vec{x}_0, L, l_0, l_F, \trans \rangle$ where
		\begin{enumerate}[(i)]
			\item
			$W = (\Omega, \calF, \mathbb{P})$ is a probability space.
			\item
			$X = \{x_1, x_2,\dotsc , x_n\}$ are the program variables 
			and 
			$R = \{r_1, r_2, \dotsc, r_m\}$ are the random variables defined on the probability space $W$.
			{Each random variable $r_i$ represents a sampled update from a random distribution, e.g.~$r_1 \sim N(0, 1)$. We assume that all $x_i$ are updated deterministically --- based on the current values of the $r_i$'s and $x_j$'s.}
			
			\item
			%$D_0$ is the initial distribution of the program variables such that all moments exist.
			$\vec{x}_0 \in \mathbb{R}^n$ is the vector of initial values of the program variables $X = \{x_1, x_2,\dotsc , x_n\}$.
			
			\item $L$ is a finite set of (program) \emph{locations};
			$l_0 \in L$ is the initial location and $l_F \in L$ is the final location. 
			{The program terminates once $l_F$ is reached.}
			\item 
			$\trans = \{\tau_1, \tau_2, \dotsc, \tau_p\}$ is a finite set of transitions.
			Each $\tau \in \trans$ is a quadruple $\langle \sloc{\tau}, \phi_{\tau}, f_{\tau}, \dloc{\tau} \rangle$ where $\sloc{\tau}$ is the source location{, which is never $l_F$,} and $\dloc{\tau}$ the destination location; $\phi_{\tau}$ is the {\em guard} \textit{assertion}, which is a boolean formula over polynomial inequalities over $X$; 
			and $f_{\tau}: \mathbb{R}^{n} \times \mathbb{R}^{m} \rightarrow \mathbb{R}^{n}$, the {\em update function}, which takes values of program variables and random samples $(\vec{x}, \vec{r}) \in \mathbb{R}^{n} \times \mathbb{R}^{m}$
			and returns the new values $\vec{x} \in \mathbb{R}^n$ for the program variables.
			The update functions are polynomials.
		\end{enumerate}
	\end{definition}
	
	A \emph{configuration} of a PTS $\Pi$ is a pair $(l, \vec{x})$ where $l \in L$ is the current program location, and $\vec{x} \in \mathbb{R}^n$ represents the contents of the program variables.
	Moreover, we say that a transition $\tau = \langle \sloc{\tau}, \phi_{\tau}, f_{\tau}, \dloc{\tau} \rangle \in \trans$ is \textit{enabled} at a configuration $(l, \vec{x})$ if $l = \sloc{\tau}$ and $\vec{x} \models \phi_{\tau}$.
	For $l \in L$, we write $\trans_l := \{ \tau \in \trans \mid \sloc{\tau} = l\}$ for the set of transitions with $l$ as the source location.
	We assume that PTSs are \emph{non-demonic}:
	For all $l \in L$,
	$\bigwedge_{\tau \not= \tau' \in \trans_l}(\phi_{\tau} \wedge \phi_{\tau'}) = \mathit{false}$ and
	$\bigvee_{\tau \in \trans_l} \phi_\tau = \mathit{true}$.
	This amounts to \emph{determinacy} in the sense that for all $l \in L$ and $\vec{x} \in \mathbb{R}^n$, there is a unique transition that is enabled at $(l, \vec{x})$.
	
	\paragraph{{\bf PTS computation.}} %The computation of a PTS $\Pi$ proceeds as follows.
	The initial configuration of 
    a PTS $\Pi$ 
    is $(l_0, \vec{x}_0)$ for some choice of initial values $\vec{x}_0 \in \mathbb{R}^n$. 
	Suppose that $\Pi$ is at a configuration $(l, \vec{x})$ at step $k$. 
	Let $\tau$ be the unique transition that is enabled (i.e.~$l = \sloc{\tau}$ and $\vec{x} \models \phi_\tau$), and $\vec{r} \in \mathbb{R}^m$ be a fresh sample drawn from the respective random distributions.
	Then the configuration of $\Pi$ at step $k+1$ is $(\dloc{\tau}, f_{\tau}(\vec{x}, \vec{r}))$.
	
	%%%%%%%%%%%%%%%%%%%%%%%%%%%%%%%%%%%%%%%%%%%%%%%%%%%%%%%%%%%%%%%%%%%%%%%%%%%%%%%%%%%%%%
	% OPERATIONAL SEMANTICS %
	%%%%%%%%%%%%%%%%%%%%%%%%%%%%%%%%%%%%%%%%%%%%%%%%%%%%%%%%%%%%%%%%%%%%%%%%%%%%%%%%%%%%%%
		
	%In this paper we only consider non-demonic transition systems. 
	The operational semantics of a PTS $\Pi$ is given by a sequence of
	random variables $\{ L^{[k]}, X^{[k]} \}$ on $W$ representing the distributions of the location and the (contents of the) program variables at each computation step $k$. 
	{The runtime of the program is the first time the final location $l_F$ is reached, i.e.~$T = \inf \{k \in \mathbb{N} :  L^{[k]} = l_F \}$.}
    We use a very similar framework as in e.g.~\cite{a,b,h,costanalysis}; details can be found in Appendix~\ref{sec: operational semantics}.

	\paragraph{\bf Probabilistic invariants.}  
    As invariants of PTS we consider functions on the program variables~\cite{wp1} and a step counter\footnote{
    In practice it is very useful to allow the runtime to be part of an invariant; see Example~\ref{eg: running example computing martingales}~\&~Figure~\ref{fig:running-example-linear-loop}.}
    such that their expected value remains constant over the computation. 
    By the pre-expectation their expected value of the next step can be expressed by the current values of the program variables~\cite{wp1,DBLP:conf/qest/GretzKM13,a,b,d,h,DBLP:journals/pacmpl/HarkKGK20,DBLP:conf/sas/KatoenMMM10}. 

    \begin{definition}[Pre-Expectation]\rm
		\label{def: pre-expectation}
		Let $\Pi$ be a PTS.
		Let $\tau \in \trans$ and $h: L \times \mathbb{R}^n \times \mathbb{N} \rightarrow \mathbb{R}$ be a measurable function over locations, valuations of the program variables, and the step counter.
		%\\
		The {\em pre-expectation} pre$\mathbb{E}(h, \tau) : \mathbb{R}^n \times \mathbb{N} \rightarrow \mathbb{R}$ is the function
		\[\textstyle
		(\vec{x}, k) \mapsto  
		\mathbb{E}_R \left(  h(\dloc{\tau}, f_{\tau}(\vec{x}, R), k) \right)
		\]
		where $\mathbb{E}_{R}$ is the expectation w.r.t.~the distribution of the random variables $R = \{r_1, \cdots, r_m \}$. 
		The pre-expectation\index{pre-expectation} for the PTS, pre$\mathbb{E}(h): L \times \mathbb{R}^n \times \mathbb{N} \rightarrow \mathbb{R}$, is given by 
		\[
            \left( l, \vec{x}, k \right) \mapsto \sum_{\tau \in \trans} [(l = \sloc{\tau}) \wedge (\vec{x} \models \phi_\tau)] \cdot \mathrm{pre}\mathbb{E}(h,\tau)(\vec{x},k)
		\]
        where $[{-}]$ is the Iverson bracket.
	\end{definition}

	\begin{remark}\label{rem: replacing powers} 
        If all the update functions $f_{\tau}$ are polynomials and $h$ is polynomial for all $l \in L$, 
        then 
		%Note that 
        %$\mathrm{pre}\mathbb{E}(h)(L^{[k]}, X^{[k]}, k)$ is computed by replacing powers of random variables~$r_i^a$ by their respective moments $\mathbb{E} \left( r_i^a \right)$ in the expression $h(\dloc{\tau}, f_{\tau}(\vec{x}, R), k)$ for each $\tau \in \mathbb{T}$. 
        $\mathrm{pre}\mathbb{E}(h, \tau)(\vec{x}, k)$ is computed by replacing powers of random variables~$r_i^a$ by their respective moments $\mathbb{E} \left( r_i^a \right)$ in the expression $h(\dloc{\tau}, f_{\tau}(\vec{x}, R), k)$. 
	\end{remark} 

	%Intuitively, given a current value of a program variable $x$, the current step-counter $k$, and a program position $l$, the pre-expectation returns us the expected value of the function $h$ after executing one more step in the program. 
    %It only relies on the information known before this program step is executed.

    %For a fixed location $l \in L$, 
	We consider \emph{probabilistic invariants} as location-independent polynomials $p:\mathbb{R}^{n} \times \mathbb{N} \rightarrow \mathbb{R}$ such that the random variables $\{p(X^{[k]}, k)\}_{k \in \mathbb{N}}$ (see Definition~\ref{def: operational semantics PTS}) are 
    %(sub/super)martingales, i.e. 
	step-wise expectation invariant: 
    \begin{align}
        \label{eg: definition PTS invariant}
         \mathrm{pre}\mathbb{E}(p)(l,\vec{x}, k+1) = p(\vec{x}, k) \qquad \mathrm{for \ all \ } l\neq l_F, \vec{x} \in \mathbb{R}^n \mathrm{ \ and \ } k \in \mathbb{N}. 
    \end{align}

    Contrary to the deterministic case, probabilistic step-wise invariants are not necessarily global invariants. 
    As a warning example, consider the following loop. 

    \begin{example}[Martingale betting strategy]%[Fair gambling]
	 	\label{example: gambling}
	 	The gambler starts by wagering 1 unit on a fair bet. 
	 	So long as she loses, she continues by wagering twice the most recent bet on the next play.
	 	Thus, when she eventually (and inevitably) wins, the amount won will cover all her previous losses and profit her by one unit.
	
	 	\medskip
	
	 	\noindent\begin{minipage}{0.4\textwidth}
	 		\centering
	 		\begin{algorithmic}
	 			\STATE {$x_1 := 1$}
                \STATE {$x_2 := 0$} 
	 			%\STATE {$x_2 := 0$} 
	 			\WHILE {$x_2 \leq 0$}
	 			\STATE {{$r_1 \sim \mathrm{Bernoulli}(\frac{1}{2})$}}
	 			\IF {$r_1 \leq 0.5$}
	 			\STATE {$x_2 := x_2 - x_1$}
	 			\ELSE
	 			\STATE {$x_2 := x_2 + x_1$}
	 			\ENDIF
	 			\STATE {$x_1 := 2 \cdot x_1$}
	 			\ENDWHILE
	 		\end{algorithmic}
	 	\end{minipage}%
	 	\begin{minipage}{0.6\textwidth}	
	 		Consider the while loop
	 		where {$x_1$ (current bet) and $x_2$ (her stake thus far)} are the program variables and $r_1$ is the random variable.

            \smallskip

            The expression
	 		$M_k := \nstep{x_2}{k}$ gives a step-wise invariant for this loop. 
            In particular, $\mathbb{E}(M_k) = 0$ holds for all $k \in \mathbb{N}$.

            \smallskip

            However, note that $\mathbb{E}(M_0) = 0 < 1 = \mathbb{E}(M_T)$, 
            where $T := \inf\{k \in \mathbb{N} \mid \nstep{x_2}{k} > 0 \}$ is the exit time of the loop.
            Hence, $(M_k)_{k \in \mathbb{N}}$ is not a loop invariant. 
	 	\end{minipage}		
	\end{example}

    Soundness of step-wise invariants can be guaranteed by the Optional Stopping Theorem (OST) from probability theory, using the 
    %This requires the 
    concept of martingales.

    %%%%%%%%%%%%%%%%%%%%%%%%%%%%%%%%%%%%%%{}%%%%%%%%%%%%%%%%%%%%%%%%%%%%%%%%%%%%%%%
	% Mathematical background %
	%%%%%%%%%%%%%%%%%%%%%%%%%%%%%%%%%%%%%%%%%%%%%%%%%%%%%%%%%%%%%%%%%%%%%%%%%%%%%%
	
	\section{Martingales and the Optional Stopping Theorem} \label{sec: background}
	
	%\subsection{Notations and basic definitions}
	We briefly review standard notions in probability theory (see e.g.~\cite{grimmett,Williams91,durrett}) and fix notations along the way\footnote{
        In Appendix~\ref{section: appendix martingales} we elaborate on these abstract notions by studying some concrete examples. 
        %In particular, we elaborate on how filtrations can be used to represent information.
    }. 
	%Henceforth 
    We fix a probability space $(\Omega, \calF, \mathbb{P})$. 
    We use filtrations to represent information known over time. 
	Formally, a \emph{filtration} is a nested, finite or infinite sequence $\{\calF_k\}_{k \in \mathbb{N}}$ of $\sigma$-algebras over $\Omega$ such that  
    %$\forall k \in \mathbb{N} \, . \, \calF_k \subseteq \calF_{k+1} \subseteq \calF$ for all $k \in \mathbb{N}$.   
    $\calF_k \subseteq \calF_{k+1} \subseteq \calF$ for all $k \in \mathbb{N}$. 
	Intuitively, the $\sigma$-algebra $\calF_k$ represents the information that is known at time $k \in \mathbb{N}$.
	A sequence of random variables $\{X_k\}$ is \emph{adapted} to a filtration $\{\calF_k\}_{k \in \mathbb{N}}$ 
    if $X_k$ is $\mathbb{B}(\mathbb{R})$-$\calF_k$-measurable for all $k \in \mathbb{N}$, i.e.~for all open $A \subseteq \mathbb{R}$ we have 
    $X_k^{-1}(A) \in \calF_k$. 
	A random variable $T: \Omega \rightarrow \mathbb{N} \cup \{\infty\}$ is called a \emph{stopping time} with respect to a filtration $\{\calF_k\}_{k \in \mathbb{N}}$ if 
	$\{T \leq k\}\in \calF_k$ for all $k \in \mathbb{N}$. 
    %{We denote by {$\text{supp}(X) = \{ \omega \in \Omega : X(\omega) \neq 0\}$} the support of the random variable $X$ which is bounded when $\sup \{|X(\omega)|\mid \omega \in \text{supp}(X)\}<\infty$, and unbounded otherwise.}
    Moreover, we denote by $\text{supp}(X)$ the smallest closed set $S$ such that $\mathbb{P}(X\in S)=1$ the support of the random variable $X$ which is bounded when $\sup \{|x|\mid x \in \text{supp}(X)\}<\infty$, and unbounded otherwise.
    
    In this paper we consider the filtration $\{\calF_k\}_{k \in \mathbb{N}}$ such that $\calF_k$ contains the information of the initial value $X^{[0]}$ and the first $k$ random samples (see Appendix~\ref{sec: operational semantics}).  
	Let $X$ be an integrable random variable, that is 
    $\mathbb{E} \left[ |X| \right] < \infty$.
    %$\int |X| d \mathbb{P} < \infty$. 
    Further, let  $\calG$ be a $\sigma$-subalgebra of $\calF$. 
	The \emph{conditional expectation} $\mathbb{E}[X \mid \calG]$ is a $\calG$-measurable function s.t.~\(
	%\int_{A} X \ d \mathbb{P} = \int_{A} \mathbb{E}[X \mid \calG] \ d \mathbb{P}
    \mathbb{E} \left[1_A \cdot X \right] = \mathbb{E} \left[1_A \cdot \mathbb{E}[X \mid \calG] \right]
	\) for all $A \in \calG$.
	The conditional expectation is $\mathbb{P}$-almost surely ($\mathbb{P}$-a.s.) unique, linear and monotone in $X$\footnote{
        See Theorem~\ref{thm: properties conditional expectation} for more properties and computation rules of the conditional expectation. 
    }. 

    ~

	Given a filtration $\{\calF_k\}$, an integrable, $\{\calF_k\}$-adapted sequence of random variables $\{M_k\}$ is 
	\begin{enumerate}[(i)]
		\item a {\em martingale} if for all $k \in \mathbb{N}$,
		$\mathbb{E}[M_{k+1} \mid \calF_k] = M_k$ almost surely (a.s.);
		\item a	{\em supermartingale} if for all $k \in \mathbb{N}$,
		$\mathbb{E}[M_{k+1} \mid \calF_k] \leq M_k$ a.s.;
		\item a	{\em submartingale}	if for all $k \in \mathbb{N}$,
		$\mathbb{E}[M_{k+1} \mid \calF_k] \geq M_k$ a.s..
	\end{enumerate}

    {In this paper we will consider martingales as the probabilistic counterpart to (algebraic) invariants, which capture that the expected value of a given expression stays invariant throughout iterations. For a martingale $\{M_k\}$ to be an invariant, one would naturally require that the expected value after loop termination is equal to its initial expected value, i.e.~$\mathbb{E}(M_T)=\mathbb{E}(M_0)$. Because the stopping time $T$ is itself a random variable, this property unfortunately is not true in general, as shown in Example~\ref{example: gambling}. Difficulties generally arise when the tails of $T$ are too heavy, or the updates are not well-behaved. 
	However, by imposing suitable boundedness conditions on $T$ or the martingale, the equality $\mathbb{E}(M_T) = \mathbb{E}(M_0)$ follows from the  Optional Stopping Theorem (OST).
    %-- OST~\cite[Theorem~5.3.1]{cohen}.
    
    There are numerous versions of the OST in the literature, which differ in the precondition that is imposed (e.g.~\cite[\S 10.10]{Williams91}, \cite{costanalysis}), for a comprehensive list see Appendix~\ref{sec: appendix ost}.
	The OST applies, for instance, if there exists a constant upper bound for the runtime~\cite{Williams91}, or if the step-wise differences are bounded by a constant and the expected value of the runtime is finite~\cite{grimmett}. 
    Because bounding the steps or values by a constant is often not possible, alternative preconditions have been developed~\cite{wang2021,costanalysis}, accommodating bounds that are expressed as functions of the iteration counter $k$. 
    In~\cite{wang2021}, the authors show that if $|M_k|\leq p(k)$ almost surely (where $p(k)$ is non-decreasing) and $\mathbb{E}(p(T))<\infty$, then the  requirements for the OST are fulfilled. The almost-sure bound, however, prevents the application to martingales where the differences depend on unbounded-support distributions. Since in that case there always is a nonzero probability that a deterministic bound is exceeded, we allow the bound to be a random variable in the sense that it may depend on a specific $\omega\in\Omega$, rather than being uniform over all execution traces. This generalises previous approaches, which require a bound for the martingale(-differences) independent of $\omega$.
    
    A general precondition~\cite[Theorem 4.6.7]{cohen} for the OST requires $\{M_k\}$ to be uniformly integrable. In general uniform integrability is notoriously hard to check, but in case the stopped martingale can be bounded by an integrable random variable, it holds true. We formulate this as a precondition for the OST in Theorem~\ref{thm: ost}. 

    \begin{restatable}[Optional Stopping and \dui]{theorem}{thmOptionalStopping}
		\label{thm: ost}
		Let $(\Omega, \calF, \mathbb{P})$ be a probability space, and $\{M_k\}_{k \in \mathbb{N}}$ be a (super)martingale adapted to a filtration $\{\calF_k\}_{k \in \mathbb{N}}$, and $T: \Omega \rightarrow \mathbb{N} \cup \{\infty\}$ a stopping time.
		If the following precondition holds:
		\begin{description}
			\item[(\dui)]\label{dui} \emph{{\bf d}ominating {f}unction for {\bf u}niform {\bf i}ntegrability:} 
            there is a random variable $X : \Omega \rightarrow \mathbb{R}_+$ such that $\mathbb{E}(X) < \infty$ and $\sup_{ k \in \mathbb{N}} |M_{\min(k, T) }| \leq X$ almost surely, %\footnote{{See Appendix~\ref{proof idf} for a proof.}}
		\end{description}
		then $M_{T}$ is integrable and $\mathbb{E}(M_{T}) = \mathbb{E}(M_0)$, if $M$ is a martingale; 
        resp.~$\mathbb{E}(M_{T}) \leq \mathbb{E}(M_0)$, if $M$ is a supermartingale.
    \end{restatable}

    \emph{We apply Theorem~\ref{thm: ost} to reason about invariants synthesised from probabilistic programs}.
	The random variable $M_k$ in Theorem~\ref{thm: ost} is an expression built up from the program variables and locations in the $k$-th step of the computation. 
	Thus, the  {\dui} precondition requires that the value of $\sup_{k \in \mathbb{N}} |M_{\min(k, T)}|$ is bounded above by an integrable random variable $X$. The random variable $X$ can be seen as a function mapping from the input variables and the actual values of the samples drawn to a real number. Because $X$ is a random variable instead of a deterministic function of the iteration counter, we can overcome the limitations of~\cite{wang2021} with respect to unbounded-support distributions.}

    \begin{remark}
        \label{remark: infinite stopping time}
        %In Theorem \ref{thm: ost} we prove  that $\{ M_{ \min(n, T) } \}_{n \in \mathbb{N}}$ is uniformly integrable. 
        {
        In Theorem~\ref{thm: ost} we do not assume that the stopping time $T$ is finite. 
        In the proof we show that $\{ M_{ \min(k, T) } \}_{k \in \mathbb{N}}$ is uniformly integrable.
        It follows from probability theory, e.g.~\cite{klenke}, that on the event of diverging runs $\{ T = \infty\}$, the martingales $M_k$ converge to a limit value $M_{\infty}$ (which is an integrable random variable).
        Hence, the random variable $M_T$ is in fact defined as \[
            M_T = \lim_{k \rightarrow \infty} M_{\min(k, T)} 
            = M_T \cdot 1_{\{ T < \infty \}} 
            + M_{\infty} \cdot 1_{\{ T = \infty \}}.
            \]
        We can illustrate this by the following example. 	
        
        \noindent\begin{minipage}{0.4\textwidth}
	 		\centering
        \begin{algorithmic}
	 			\STATE {$x := 1$}
	 			\WHILE {$x \geq 0$}
	 			\STATE {$x := 1$}
	 			\ENDWHILE
	 	\end{algorithmic}
        	\end{minipage}%
	 	\begin{minipage}{0.6\textwidth}	
            Clearly, $\mathbb{P}\left(T = \infty\right) = 1$, where $T := \inf\{k \in \mathbb{N} : x^{[k]} < 0\}$ is the exit time of the loop. 
            Nonetheless, $\{x^{[k]}\}_{k \in \mathbb{N}}$ is a martingale, hence an invariant, for this loop and we have $\mathbb{E}(x^{[T]}) = 1 = \mathbb{E} (x^{[0]})$.
	 	\end{minipage}}
        %Hence, Theorem~\ref{thm: ost} also holds if the stopping time $T$ is not almost surely finite.
    \end{remark}
	%\smallskip

    %\begin{remark} 
%        In the setting of Theorem~\ref{thm: ost} and by \cite[Theorem~10.9]{Williams91}, $\mathbb{E} [ M_{\min(k, T)} ] = \mathbb{E} \left[ M_{0} \right]$ for all $k \in \mathbb{N}$. 
%        So, it is natural to assume that $\mathbb{E} [ M_{T} ] = \mathbb{E} \left[ M_{0} \right]$,  
%        which can be false by Example~\ref{example: gambling}.
%        %We have seen in Example~\ref{example: gambling} that this is not always the case. 
%        Generally when intuitive result fail in probability theory, then because something is getting too large. 
%        In Example~\ref{example: gambling}, we have $x_2^{[\min(k, T)]} = 1$, if $k \leq T$, and $x_2^{[\min(k, T)]} = -(2^{k} - 1)$ otherwise. 
%        Hence, $\mathbb{E} [ \sup_{k \in \mathbb{N}} |x_2^{[\min(k, T)]}| ] = \sum_{k = 1}^{\infty} (1/2)^{k-1} \cdot (2^{k} - 1)  = \infty$. 
%    \end{remark}

    {While Theorem~\ref{thm: ost} is general and does not directly give rise to automatically verifiable preconditions, we use it to derive Theorem~\ref{thm:moment-monomials}. 
    In the case of programs with linear updates without exponential growth this theorem gives novel preconditions that are much more tractable, forming the basis of our implementation.}

	\paragraph{\bf Probabilistic martingale invariants.} 

    By the following lemma our (step-wise) probabilistic invariants are martingales.
    %By the following lemma we can construct martingales via the pre-expectation. 
    This is proven in Appendix~\ref{sec: operational semantics}.

    \begin{restatable}[Invariants and Pre-Expectations]{lemma}{lemmatilde}
		\label{lemma tilde}
		Let $\Pi$ be a PTS with operational semantics $\{L^{[k]}, X^{[k]}\}_{k \in \mathbb{N}}$, and 
        let $\calF_k$ be the smallest $\sigma$-algebra such that the initial distribution and the first $k$ random samples are $\calF_k$-measurable.  
        Let $h: L \times \mathbb{R}^n \times \mathbb{N} \rightarrow \mathbb{R}$ be such that for all $l \in L$, $h(l, -, -) : \mathbb{R}^n \times \mathbb{N} \rightarrow \mathbb{R}$ is a polynomial.
		\begin{enumerate}[(i)]
			\item\label{item: lemma tilde i} 
            Then $\mathbb{E}[h(L^{[k+1]}, X^{[k+1]}, k+1) \mid \calF_k] = \mathrm{pre}\mathbb{E}(h)(L^{[k]}, X^{[k]}, k+1)$ for all $k \in \mathbb{N}$.
			
			\item\label{item: lemma tilde ii} If $\mathrm{pre}\mathbb{E}\left(h \right)(l, \vec{x}, k+1) \leq h(l, \vec{x}, k)$ for all $l \in L, \vec{x} \in \mathbb{R}^n, k \in \mathbb{N}$,		
			then $\{h(L^{[k]}, X^{[k]}, k)\}_{k \in \mathbb{N}}$ is a supermartingale adapted to $\{\calF_k\}_{k \in \mathbb{N}}$.
		\end{enumerate}
	\end{restatable}

    \begin{remark}
        \label{rem: stopped martingale} 
        By e.g.~\cite[Theorem~10.9]{Williams91}, if $(M_n)_{n \in \mathbb{N}}$ is a martingale and $T: \Omega \rightarrow \mathbb{N} \cup \{\infty\}$ is a stopping time, then $(M_{\min(n, T)})_{n \in \mathbb{N}}$ is a martingale. 
       Hence, we do not need to consider the loop guard when computing loop invariants. 
       We will use this fact throughout the rest of the paper.
    \end{remark}

	\section{Polynomial Invariants of Linear Programs}
	\label{sec: linear}
	
	This section focuses on the class of linear probabilisitc transition systems (PTS) and shows that the {\dui} precondition can be used to obtain non-linear invariants of linear PTS. 
	We provide a sufficient condition under which polynomial martingale expressions satisfy the Optional Stopping Theorem (OST) for programs consisting of a single loop (Theorem~\ref{thm:moment-monomials}). 
    For the remainder of this paper Definition~\ref{def:triangular-linear-update-loop} imposes additional restrictions on the updates. 

    %We denote by $\vec{x}^{[k]}$, resp.~$\vec{r}^{[k]}$, the values of the program variables and random samples in the $k$-th step of the computation.

    \begin{definition}[Single-Loop Linear Programs]
    \label{def:triangular-linear-update-loop}
        A \emph{simple single-loop linear program} consists of a vector of variables $\vec{x}$, which is updated in the loop body through the following update function:
        \begin{align}
        \label{eq: linear update}
            \vec{x}^{[k+1]}=A\vec{x}^{[k]}+B\vec{r}^{[k+1]}
        \end{align}
        where $A \in \mathbb{R}^{n \times n}$, $B \in \mathbb{R}^{n \times m}$ and $\vec{x}^{[k]}$ denotes the value of $\vec{x}$ in iteration $k$. 
        %We impose that the largest eigenvalue of $A$, which we call $\lambda_\text{max}$, satisfies $|\lambda_\text{max}|\leq1$. 
        We impose that all eigenvalues $\lambda$ of $A$ satisfy $|\lambda| \leq 1$.
        Moreover, $\vec{r}^{[k+1]}\sim\vec{R}$ is a vector of fresh draws from random distributions,  
        %(for the $k+1$-th step), 
        where all moments are finite, i.e.~$\forall i\in\{1,2,\dots,m\},\forall N\in\mathbb{N}\mathpunct.\mathbb{E}(|R_i|^N)<\infty$. 
        Samples are assumed to be independent.
    \end{definition}

    The interactions between variables must be represented through a matrix with eigenvalues with absolute value at most~$1$. Current techniques~\cite{wang2021} for verifying the applicability of the OST work for those programs, when random updates sample only from distributions with bounded support. Importantly, we allow sampling from distributions with unbounded support, only requiring that all moments are finite for those distributions. Examples for such programs are given in Figures~\ref{fig:running-example-linear-loop}~\&~\ref{fig:example-unbounded-update-for-lg}.
    %--- see also Example~\ref{ex:pts:rec:def}. 
    While the program class is restrictive, our approach can be adapted to loops with branching, by overapproximating the updates made in each iteration.

    \begin{example}\label{ex:pts:rec:def}
        The loop in Figure~\ref{fig:running-example-linear-loop} is a program as in Definition~\ref{def:triangular-linear-update-loop} with an update matrix with eigenvalues $1$. Concretely we can formulate the following recurrence: $$\begin{pmatrix}
            z^{[k+1]}\\
            y^{[k+1]}\\
            x^{[k+1]}\\
            k^{[k+1]}\\
            1\\
        \end{pmatrix}=\begin{pmatrix}
            1 & 1 & 0&0&0\\
            0&1&0&0&0\\
            0&0&1&0&0\\
            0&0&0&1&1\\
            0&0&0&0&1
        \end{pmatrix}\begin{pmatrix}
            z^{[k]}\\
            y^{[k]}\\
            x^{[k]}\\
            k^{[k]}\\
            1\\
        \end{pmatrix}+\begin{pmatrix}
            1&1&1\\
            1&1&0\\
            1&0&0\\
            0&0&0\\
            0&0&0
        \end{pmatrix}\begin{pmatrix}
            u^{[k+1]}\\
            n_1^{[k+1]}\\
            n_2^{[k+1]}
        \end{pmatrix}.$$
        The variables $u^{[k+1]}, n_1^{[k+1]}$ and $n_2^{[k+1]}$ are fresh samples from $\mathrm{Uniform}(-1,0)$, $\mathrm{Normal}(1,2)$, and $\mathrm{Normal}(-2,4)$ respectively.
    \end{example}

    Moreover, linear loops are guarded by a loop guard $\phi$ and as such, in what follows, the runtime $T := \inf\{k\in\mathbb{N}\mid \lnot \phi(\vec{x}^{[k]})\}$ is defined as the first time this condition is violated. We furthermore write $x\in \phi$, when the program variable $x$ occurs in $\phi$ syntactically.

    {We next show how dominating functions (up to the runtime $T$) can be constructed for monomials over the program variables. As integration is a linear operator, 
    %it distributes over addition and 
    a polynomial martingale candidate is uniformly integrable when the same is true for its individual monomials (e.g.~\cite[Theorem 6.18]{klenke}). Since the construction of the dominating function for a monomial depends on the structure of the Jordan decomposition {of the update (sub)matrix of} the program variables it depends on, we first make the notion of dependence precise in Definition~\ref{def:dependence}.}

    \begin{definition}[Variable Dependence]
        \label{def:dependence}
        For two program variables $x_i, x_j$, % \in \{x_1, \cdots, x_n\}$,  
        we write $x_i\preceq x_j$ when \emph{$x_i$ depends on $x_j$}, which is the case when $A_{ij}\neq 0$.
        Similarly, a program variable $x_i$ depends on a probability distribution $R_j$, when $B_{ij}\neq 0$.
        Furthermore, $\preceq$ is transitive and reflexive.
    \end{definition}

    %As a notation 
    Let $I_i=\{j\mid x_i\preceq x_j\}$ be the indices of the variables that the variable $x_i$ depends on. We denote by $A_i$ the sub-matrix of $A$, consisting only of the rows and columns with indices in $I_i$. Similarly, $B_i$ is the sub-matrix of $B$ consisting only of the rows with index in $I_i$. By definition of dependence, only considering $A_i,B_i$ instead of $A,B$ has no effect on the value of $x_i$.

    In order to apply the {\dui} precondition, for a fixed program variable $x_i$ we first need to construct an (integrable) random variable $Z_i$ that is greater than $|x_i^{[\min(k,T)]}|$ almost surely {for all $k \in \mathbb{N}$} simultaneously. The following example illustrates how such a random variable can be constructed, while Theorem~\ref{thm:dominating_function_for_alpha_1}  establishes that for our programming model (Definition~\ref{def:triangular-linear-update-loop}) this is always possible.

    \begin{example}\label{ex:bound}
        Consider the program in Figure~\ref{fig:running-example-linear-loop}. We denote by $u^{[k]}$, $n_1^{[k]}$ and $n_2^{[k]}$ the random samples drawn from $\text{Uniform}(-1,0)$, $\text{Normal}(1,2)$ and $\text{Normal}(-2,4)$ in the $k$-th iteration. We denote by $u^{[k]}(\omega)$ and $n_i^{[k]}(\omega)$ the value for a given element from the sample space $\omega \in \Omega$. 
        %Using those samples, 
        {Then} we construct a random variable $Z_y:\Omega\to\mathbb{R}$, such that $\forall\omega\in\Omega, \forall k \in \mathbb{N} : |y^{[\min(k,T) ]}(\omega)| \leq Z_y(\omega)$.

        By Figure~\ref{fig:running-example-linear-loop}, 
        %the definition of the program 
        $y^{[k]}(\omega)=y^{[k-1]}(\omega) + n_1^{[k]}(\omega) + u^{[k]}(\omega)$. 
        Hence, we arrive at the closed form expression $y^{[k]}(\omega) = y^{[0]}+ \sum_{i=1}^{k} (n_1^{[i]}(\omega) + u^{[i]}(\omega))$.
        Because the absolute value is sub-additive, $|y^{[k]}(\omega)| \leq |y^{[0]}|+ \sum_{i=1}^{k} (|n_1^{[i]}(\omega)|+|u^{[i]}(\omega)|)$. 
        Since that bound is monotonically increasing, we can choose
        $Z_y(\omega) := |y^{[0]}|+ \sum_{i=1}^{T(\omega)} (|n_1^{[i]}(\omega)|+|u^{[i]}(\omega)|)$ 
        %it applies that $|y^{[\min(k,T)]}(\omega)|\leq |y^{[T]}(\omega)|\leq |y^{[0]}|+ \sum_{i=1}^T |r_1^{[i]}(\omega)|+|r_2^{[i]}(\omega)|:= Z_y$.
        and so 
        %have {indeed} that 
        $|y^{[\min(k, T(\omega))]}(\omega)| \leq Z_y(\omega)$ for all 
        $k \in \mathbb{N}$.
    \end{example}
    
    The construction of  bounds similar to Example~\ref{ex:bound} can be done in general --- as stated in Theorem~\ref{thm:dominating_function_for_alpha_1} and detailed next. 
    This plays a crucial role in proving Theorem~\ref{thm:moment-monomials} using Theorem~\ref{thm: ost}. 
    We will prove Theorem~\ref{thm:dominating_function_for_alpha_1} in Appendix~\ref{sec: proofs linear} (see Theorem~\ref{thm:dominating_function_for_alpha_1:appendix}).

    \begin{restatable}[Almost Sure Bounds]{theorem}{thmdominatingfunctionlinearpts}
        \label{thm:dominating_function_for_alpha_1}
        Assume that all eigenvalues $\lambda$ of $A$ satisfy $|\lambda| \leq 1$.
        For a program variable $x_i$, let $A_i, B_i$ be the update (sub-)matrices for the program variables that $x_i$ depends on.
        Let $m_i$ be the size of the largest Jordan block of the Jordan normal form of $A_i$. 
        Then there are constants $c_{i, 1}, c_{i, 2} \in \mathbb{R}$ that depend only on the matrices $A_i, B_i$, s.t.~if we define
        \begin{align*}
            Z_i:= c_{i, 1}(T^{m_i-1}+1)\|\vec{x}^{[0]}\|_2 + c_{i, 2}(T^{m_i-1}+1)\sum_{l=1}^T \|\vec{r}^{[l]}\|_2,
        \end{align*}
        %
        %$c_{i, 1} := a^{(m_i-1)}_i m_i \|P_i\|_2\|P_i^{-1}\|_2(1 + \gamma_i)$
        %with  $\delta_i := a^{(m_i-1)}_i m_i \|P_i\|_2\|P_i^{-1}\|_2$, $\beta_i=\|B_i\|_2$, $\gamma_i:=(m_i-1)^{m_i-1}$ and $a_i$ being the largest absolute value within $J_i$.
        %}
        %\begin{align*}
        %    Z_i:= \delta_i(T^{m_i-1}+\gamma_i)\|\vec{x}^{[0]}\|_2+\delta_i(T^{m_i-1}+\gamma_i)\beta_i\sum_{l=1}^T \|\vec{r}^{[l]}\|_2
        %\end{align*}
        %with  $\delta_i := a^{(m_i-1)}_i m_i \|P_i\|_2\|P_i^{-1}\|_2$, $\beta_i=\|B_i\|_2$, $\gamma_i:=(m_i-1)^{m_i-1}$ and $a_i$ being the largest absolute value within $J_i$.
        {then $0 \leq \big| x_i^{[\min(k, T)]} \big| \leq Z_i$ holds almost surely for all $k \in \mathbb{N}$.}
        %\begin{align*}
        %    0 \leq \big| x_i^{[\min(k, T)]} \big| \leq Z_i.
        %\end{align*}
        %given that the maximum eigenvalue of $A$ is $1$. 
        %Notably, $\beta_i=0$ when $x_i$ does not depend on any random variable.
    \end{restatable}
    
    In the following, we use a multi-index notation for monomials over the program variables,  
    i.e.~$\alpha \in \{(\alpha_1, \alpha_2, \cdots, \alpha_n) : \alpha_i \in \mathbb{N}\}$ and $\vec{x}^\alpha=\prod_{\alpha_i\in\alpha} x_i^{\alpha_i}$. 
    Moreover, $\left(\vec{x}^\alpha\right)^{[k]}$ denotes the value of $\vec{x}^\alpha$ in iteration $k$.

    \begin{restatable}[Random Variables with Unbounded Support]{theorem}{theoremuniformintegrabilitylts}
    \label{thm:moment-monomials} 
        Assume the setting of Definition~\ref{def:triangular-linear-update-loop}. 
        For all program variables $x_i$, let $A_i, B_i$ be the update (sub-)matrices for the variables that $x_i$ depends on.  
        Let $m_i$ be the size of the largest Jordan block of the matrix $A_i$ and let $\beta_i=\|B_i\|_2$. 
        Fix a multi-index $\alpha = (\alpha_1, \dots, \alpha_n)$ and set
        \[
            N := N(\alpha) = \sum_{i = 1}^n (m_i-1+1_{\{\beta_i>0\}})\alpha_i.
        \]
        Then the random variable sequence $\{(\vec{x}^\alpha)^{[\min(k,T)]}\}_{k\in\mathbb{N}}$ is uniformly integrable if 
        \begin{enumerate}[(i)]
            \item $\mathbb{E}(T^N)<\infty$ and $x\in\phi \land x\preceq R_j\implies \text{supp}(R_j)\text{ is bounded}$, or
            \item $\mathbb{E}(T^{N+1})<\infty$.
        \end{enumerate}
    \end{restatable}	
    
    For a fixed multi-index $\alpha$, 
    we use the variables $Z_i$ from Theorem~\ref{thm:dominating_function_for_alpha_1} to construct an almost sure upper bound for the monomial $(\vec{x}^\alpha)^{[\min(k, T)]}$ over all $k \in \mathbb{N}$. 
    In Theorem~\ref{thm:moment-monomials} we show that those monomials are uniformly integrable, when the expected value of the $N$-th power of the stopping time is finite, where $N$ depends on the update matrices, and $T$ does not depend on distributions with unbounded support. {Naturally this is the case, when variables in the loop guard $\phi$ depend (Definition~\ref{def:dependence}) only on random variables drawn from bounded support distributions.} 
    If $T$ does depend on distributions with unbounded support, instead finiteness of the $(N+1)$-th moment is required. 
    We prove Theorem~\ref{thm:moment-monomials} in  Appendix~\ref{sec: proofs linear}.

        \begin{example}\label{ex:pts:rec}
        The loop in Figure~\ref{fig:running-example-linear-loop} is a simple single-loop linear program as in Definition~\ref{def:triangular-linear-update-loop} with an update matrix with eigenvalue $1$, 
        {which we have written out in Example~\ref{ex:pts:rec:def}}.
        
%        Concretely we can formulate the following recurrence: $$\begin{pmatrix}
%            z^{[k+1]}\\
%            y^{[k+1]}\\
%            x^{[k+1]}\\
%            k^{[k+1]}\\
%            1\\
%        \end{pmatrix}=\begin{pmatrix}
%            1 & 1 & 0&0&0\\
%            0&1&0&0&0\\
%            0&0&1&0&0\\
%            0&0&0&1&1\\
%            0&0&0&0&1
%        \end{pmatrix}\begin{pmatrix}
%            z^{[k]}\\
%            y^{[k]}\\
%            x^{[k]}\\
%            k^{[k]}\\
%            1\\
%        \end{pmatrix}+\begin{pmatrix}
%            1&1&1\\
%            1&1&0\\
%            1&0&0\\
%            0&0&0\\
%            0&0&0
%        \end{pmatrix}\begin{pmatrix}
%            u^{[k+1]}\\
%            n_1^{[k+1]}\\
%            n_2^{[k+1]}
%        \end{pmatrix}.$$
%        The variables $u^{[k+1]}, n_1^{[k+1]}$ and $n_2^{[k+1]}$ are fresh samples from $\mathrm{Uniform}(-1,0)$, $\mathrm{Normal}(1,2)$, and $\mathrm{Normal}(-2,4)$ respectively.
%
 %       
        The update matrix ${A}$ is already in Jordan normal form. The variables  $x,y$ are only affected by Jordan blocks with size $1$, and have random updates, while $z$ is affected by a Jordan block of size $2$ and has random updates. 
        The variable $k$ is affected by a Jordan block of size $2$, but is not updated by the random samples. 
        Because the runtime depends only on bounded support distributions, 
        by Theorem~\ref{thm:moment-monomials} 
        %Hence, by Theorem~\ref{thm:moment-monomials} and the fact that the runtime depends only on bounded support distributions, 
        through ensuring that $E(T^N)<\infty$, every monomial $x^ay^bz^ck^d$ is uniformly integrable, when $a+b+2c+d\leq N$.
    \end{example}

    In conclusion, 
    a polynomial over random variables is uniformly integrable, if every monomial in it is. If each monomial of a polynomial martingale satisfies the prerequisites of Theorem~\ref{thm:moment-monomials}, then it satisfies the OST by the {\dui} precondition (Theorem~\ref{thm: ost}). 
    %In the next sections we will detail how this approach can be automated.

\subsection{Constructing Martingales for Linear Programs}
\label{subsection:constructing-martingales-linear-programs}

    We use Theorem~\ref{thm:moment-monomials} to identify uniformly integrable monomials up to the stopping time, which can be used for \emph{constructing martingales that enable bounding loop variables with unbounded support}. 
    By Definition~\ref{def:triangular-linear-update-loop} we represent the loop body by a linear update 
    %This can be done by representing the loop as a linear PTS of form 
    $\vec{x}^{[k+1]}={A}\vec{x}^{[k]}+{B}\vec{r}^{[k+1]}$. 
    Given a finite $N$, for which $\mathbb{E}(T^N)<\infty$ holds, we check which monomials $\vec{x}^\alpha$ are uniformly integrable. 
    Thereby we consider the size of the Jordan blocks affecting its variables and potential dependencies of the runtime from unbounded support distributions.
    %(see Theorem~\ref{thm:moment-monomials} and Example~\ref{ex:pts:rec}). 
    %Throughout this section we use $\vec{x}^\alpha_T=(\vec{x}^\alpha)^{[T]}$ to refer to the value of a monomial after the loop.

    % moved example to just after theorem 3

    ~

    By~\eqref{eg: definition PTS invariant} our goal is to find polynomials $m: \mathbb{R}^n \times \mathbb{N} \rightarrow \mathbb{R}$ such that 
    $\mathrm{pre}\mathbb{E}(m)(\cdot,\vec{x}, k + 1) = m(\vec{x}, k)$ holds for all $\vec{x}$ and $k$. 
    It then follows by Lemma~\ref{lemma tilde} that $\{ m(\vec{x}^{[k]}, k) \}_{k \in \mathbb{N}}$ is a martingale\footnote{
        Note that by Remark~\ref{rem: stopped martingale} we do not need to consider the loop guard at this stage.
    }. 

    ~

    As is shown in Example~\ref{ex:pts:rec:def}, the update matrix ${A}$ of a  PTS resembles a recurrence relation for every variable $x_i^{[k+1]} = p_{x_i}(\vec{x}^{[k]}, \vec{r}^{[k+1]})$, where $p_{x_i}$ is a (linear) polynomial over the current variable values and the current samples. 
    Hence, $q_{x_i} := \mathrm{pre}\mathbb{E}(x_i)(\cdot, \vec{x}, k + 1)$ is a (linear) polynomial over the variable values and the step counter\footnote{
        Often it is easier to create a program variable that stores the step counter, i.e.~$k$ in Figure~\ref{fig:running-example-linear-loop}.
    }. 

    \begin{example}
    \label{example-pre-expectation}
        Consider again the example in Figure~\ref{fig:running-example-linear-loop} and the monomial $kx$. 
        We get $p_x=x^{[k+1]}=x^{[k]}+u^{[k+1]}$ and $p_k=k^{[k+1]}=k^{[k]}+1$. Then $(kx)^{[k+1]}=p_kp_x$, and so 
        \begin{align*}
            q_{kx} 
            &:= \mathrm{pre}\mathbb{E}(kx) \left(\cdot, \vec{x}, k + 1  \right)
            %= \mathbb{E} \big( x^{[k+1]} \mid \mathcal{F}_k \big) 
            %= \mathbb{E} \big( (x^{[k]}+u^{[k+1]})(k^{[k]}+1) \mid \mathcal{F}_k \big) 
            = \mathbb{E}_u \big( (x+u)(k+1) \big)
            = x k + x + \mathbb{E} (u) k + \mathbb{E}(u) 
            \\
            &= k x - (1/2) k +x - (1/2), 
        \end{align*}
        where $u \sim \mathrm{Uniform}(-1,0)$ denotes a fresh random sample. 
        %where we have used that $x^{[k]}, k^{[k]}$ are $\calF_k$-measurable and $u^{[k+1]}$ is independent of $\calF_k$. 
        %%$q_{kx} = \mathbb{E}((x^\alpha)^{[k+1]}|\mathcal{F}_k)=\mathbb{E}(p_kp_x|\mathcal{F}_k)=k^{[k]}x^{[k]}-\frac{1}{2}k^{[k]}+x^{[k]}-\frac{1}{2}$.
        Similarly, by allowing more complex monomials and by replacing the moments of the random samples, we get
        $$\begin{array}{ll}
            q_k=k+1 \qquad&q_{kx}=k x - (1/2) k + x - (1/2)  \\
            q_x=x-(1/2)& q_{ky}=k y +(1/2)k + y + (1/2)\\
            q_y=y+(1/2)& q_{x^2}= x^2 - x + (1/3)\\
            q_z=z + y -(3/2)\qquad& q_{xy} = x y +(1/2) x - (1/2) y - (1/6)\\
            q_{k^2}=k^2+2 k +1\qquad& q_{y^2}= y^2 + y + (7/3).
        \end{array}$$
        %$$\begin{array}{ll}
        %    q_k=k^{[k]}+1 \qquad&q_{kx}=k^{[k]}x^{[k]}-\frac{1}{2}k^{[k]}+x^{[k]}-\frac{1}{2}  \\
        %    q_x=x^{[k]}-\frac{1}{2}& q_{ky}=k^{[k]}y^{[k]}+\frac{1}{2}k^{[k]}+y^{[k]}+\frac{1}{2}\\
        %    q_y=y^{[k]}+\frac{1}{2}& q_{x^2}=(x^{[k]})^2-x^{[k]}+\frac{1}{3}\\
        %    q_z=z^{[k]}+y^{[k]}-\frac{3}{2}\qquad& q_{xy}=x^{[k]}y^{[k]}+\frac{1}{2}x^{[k]}-\frac{1}{2}y^{[k]}-\frac{1}{6}\\
        %    q_{k^2}=(k^{[k]})^2+2k^{[k]}+1\qquad& q_{y^2}=(y^{[k]})^2+y^{[k]}+\frac{7}{3}.
        %\end{array}$$
    \end{example}

    By linearity and as the moments of the random samples are known, 
    this allows us to compute the pre-expectation for any monomial $\vec{x}^{\alpha}$, 
    %$\text{pre}\mathbb{E}(\vec{x}^\alpha)=\mathbb{E}\left(\prod_{ i=1}^{n} p_{x_i}(\vec{x}^{[k]}, \vec{r}^{[k+1]})^{\alpha_i}|\mathcal{F}_k\right)$ (by Lemma~\ref{lemma tilde}), 
    which is a polynomial to which we also refer as $q_\alpha$.
    %for a multi-index $\alpha$. 
    %Because by Lemma~\ref{lemma tilde} 
    %$\text{pre}\mathbb{E}((\vec{x}^\alpha)^{[k]})=\mathbb{E}\left(\prod_{ i=1}^{n} p_{x_i}(\vec{x}^{[k]}, \vec{r}^{[k+1]})^{\alpha_i}|\mathcal{F}_k\right)$, 

    \paragraph{\bf Martingales of linear PTS.} Recall that $\{M^{[k]}\}_{k\in\mathbb{N}}$ is a martingale, 
    whenever for all $k \in \mathbb{N}$ we have 
     $\mathbb{E}(M^{[k+1]}\mid \mathcal{F}_k) = M^{[k]}$. 
     For a polynomial $m(x_1,\dots,x_n)$ over the program variables to be a martingale,
     it must hence hold that  
    %\footnote{
     %using Lemma~\ref{lemma tilde}
     %} 
     $m(\vec{x}^{[k]}) = \mathbb{E}(m(\vec{x}^{[k+1]})\mid\mathcal{F}_k) = \mathrm{pre}\mathbb{E}(m)(\cdot, \vec{x}^{[k]}, k+1)$  
     %$\mathrm{pre}\mathbb{E}(m)(\vec{x}^{[k]}, k+1) = \mathbb{E}(m(\vec{x}^{[k+1]})\mid\mathcal{F}_k)=m(\vec{x}^{[k]})$, 
     where we have used Lemma~\ref{lemma tilde}.
     This can be captured by a system of linear equations.   

     \begin{example} 
     \label{eg: running example computing martingales - computing martingales}
        Consider the loop in Figure~\ref{fig:running-example-linear-loop}. 
        If we choose the template 
        $m(\vec{x}) =a_x x +a_{kx} k x +a_{k^2} k^2 + a_{x^2} x^2$, 
        then by the computations from Example~\ref{example-pre-expectation} and by Lemma~\ref{lemma tilde}, 
        %we have
        \begin{align*}
            \mathbb{E} \left( m(\vec{x}^{[k+1]}) \mid \mathcal{F}_k \right)
            = m(\vec{x}^{[k]})
            &+ (a_{kx} - a_{x^2}) x^{[k]} + \left( 2a_{k^2} - \frac{a_{kx}}{2} \right) k^{[k]} 
            - \frac{a_x}{2} - \frac{a_{kx}}{2}  + a_{k^2} + \frac{a_{x^2}}{3}.
        \end{align*}
        The equation $0=a_{kx}-a_{x^2}$ ensures that the coefficient of $x^{[k]}$ in $m(\vec{x}^{[k]})$, and the coefficient of $x^{[k+1]}$ in $\mathbb{E}(m(\vec{x}^{[k+1]})|\mathcal{F}_k)$ are the same.
        We repeat this for all coefficients and get a system of linear equations. One solution is 
        %$m(\vec{x}, k) = 3 (k^{[k]})^2 + 12 k^{[k]} x^{[k]} + 12 (x^{[k]})^2 + 2 x^{[k]}$, 
        $m(\vec{x}) = 3 k^2 + 12 k x + 12 x^2 + 2 x$, 
        which is thus a martingale. 
     \end{example}
     
    To construct martingales of linear PTS, we proceed as follows.
    We fix a set of monomials $\mathcal{M}$ which are uniformly integrable, and denote by $\mathcal{A} \subseteq \{(\alpha_1, \alpha_2, \cdots, \alpha_n) : \alpha_i \in \mathbb{N}\}$ the corresponding set of multi-indices, i.e.~$\mathcal{M} = \{\vec{x}^\alpha\mid \alpha\in\mathcal{A}\}$. 
    Moreover, each monomial $\vec{x}^\alpha$ gets assigned a coefficient variable $a_\alpha \in \mathbb{R}$. We define the following set of linear equations, with $[\vec{x}^\alpha]q(\vec{x})$ denoting the coefficient of $\vec{x}^\alpha$ in the polynomial~$q(\vec{x})$:

    \begin{equation}\label{eq:linear-system-of-equations}\left\{ a_\alpha = \sum_{\beta\in\mathcal{A}} a_\beta [\vec{x}^\alpha]q_\beta(\vec{x})  ~\Big\vert~ \alpha\in\mathcal{A}\right\}.\end{equation}

    Every solution of Equation~\eqref{eq:linear-system-of-equations} corresponds to a polynomial $m(\vec{x}):=\sum_{\alpha\in\mathcal{A}} a_\alpha\vec{x}^\alpha$, which satisfies the martingale property. Intuitively, this equation enforces that for all monomials, the coefficient in $m(\vec{x}^{[k]})$ is the same as in $\mathbb{E}(m(\vec{x}^{[k+1]})|\mathcal{F}_k)$.

    Combined with Theorem~\ref{thm:moment-monomials}, when these martingales satisfy the OST, the expected value after loop termination equals its initial value. 
    Hence, throughout we will use $\vec{x}^\alpha_T=(\vec{x}^\alpha)^{[T]}$ to refer to the value of a monomial after the loop.

    \begin{example}
        \label{eg: running example computing martingales}
        In  Figure~\ref{fig:running-example-linear-loop} we have $\mathbb{E}(T^2)<\infty$. 
        Hence, by Example~\ref{ex:pts:rec} we can consider the set of uniformly integrable monomials $\{x,k,kx,k^2,x^2\}$, for which we have computed $q_\alpha = \mathrm{pre}\mathbb{E}(\vec{x}^{\alpha})$ in Example~\ref{example-pre-expectation}. 
        So, by Example~\ref{eg: running example computing martingales - computing martingales}, $m(\vec{x}) = 3 k^2 + 12 k x + 12 x^2 + 2 x$ is a martingale for this loop. 
        In particular, this martingale satisfies the OST by Theorem~\ref{thm:moment-monomials} and by Example~\ref{ex:pts:rec}, i.e.
        \begin{equation}
            \mathbb{E}(3k_T^2 + 12k_Tx_T + 12x_T^2 + 2x_T) = 12x_0^2 + 2x_0
            \label{eq:martingale-for-deriving-k2}.
        \end{equation}
        %
        % $$\begin{array}{ll}
        %     \mathbb{E}(k^{[k+1]}|\mathcal{F}_k)=k^{[k]}+1 \qquad&\mathbb{E}(k^{[k+1]}x^{[k+1]}|\mathcal{F}_k)=k^{[k]}x^{[k]}-\frac{1}{2}k^{[k]}+x^{[k]}-\frac{1}{2}  \\
        %     \mathbb{E}(x^{[k+1]}|\mathcal{F}_k)=x^{[k]}-\frac{1}{2}& \mathbb{E}(k^{[k+1]}y^{[k+1]}|\mathcal{F}_k)=k^{[k]}y^{[k]}+\frac{1}{2}k^{[k]}+y^{[k]}+\frac{1}{2}\\
        %     \mathbb{E}(y^{[k+1]}|\mathcal{F}_k)=y^{[k]}+\frac{1}{2}& \mathbb{E}((x^{[k+1]})^2|\mathcal{F}_k)=(x^{[k]})^2-x^{[k]}+\frac{1}{3}\\
        %     \mathbb{E}(z^{[k+1]}|\mathcal{F}_k)=z^{[k]}+y^{[k]}-\frac{3}{2}\qquad& \mathbb{E}(x^{[k+1]}y^{[k+1]}|\mathcal{F}_k)=x^{[k]}y^{[k]}+\frac{1}{2}x^{[k]}-\frac{1}{2}y^{[k]}-\frac{1}{6}\\
        %     \mathbb{E} ((k^{[k+1]})^2|\mathcal{F}_k)=(k^{[k]})^2+2k^{[k]}+1\qquad& \mathbb{E}((y^{[k+1]})^2|\mathcal{F}_k)=(y^{[k]})^2+y^{[k]}+\frac{7}{3}.
        % \end{array}$$
        %
       Similarly by considering $\mathcal{M}=\{k,x,y,z,k^2,kx,ky,x^2,xy,y^2\}$ we get:
        \begin{align}
            \mathbb{E}(k_T + 2x_T) &= 2x_0 \label{eq:martingale-k-x}\\
            \mathbb{E}(x_T + y_T) &= x_0 + y_0\label{eq:martingale-x-y}\\
            \mathbb E(14x_T + 3 x_T^2 + 6 x_Ty_T + 3 y_T^2)&=14 x_0 + 3 x_0^2 + 6 x_0y_0 + 3 y_0^2\label{eq:martingale-for-square-extraction}
        \end{align}  
    \end{example}

    \subsection{Deriving Bounds for Moments}
     {In Section~\ref{subsection:constructing-martingales-linear-programs} we derived martingale expressions for linear PTS with a single loop that satisfy the OST. These martingales, together with the negated loop guard and a bound on its value after termination, can be used to derive bounds for moments of program variables after termination. For that, we implement a saturation algorithm that systematically applies predefined rules in order to \emph{automatically derive bounds on the expected value of monomials after PTS termination}, as follows. While our rules exploit  well-known mathematical properties, to the best of our knowledge our symbolic propagation approach has not been implemented elsewhere.}

    The computation of bounds works by maintaining a set of bounds of monomials $\vec{x}_T^\alpha$ and their expected value $\mathbb{E}(\vec{x}_T^\alpha)$ at the stopping time, and using a set of derivation rules to derive new bounds. The set of derivation rules consists of predefined rules, which apply in general to (random) variables, as well as rules which are obtained from the previously generated martingales. We implemented such derivation of bounds in extension to the probabilistic loop analysis tool \texttt{Polar}~\cite{DBLP:moosbrugger-polar}, allowing us to automatically derive bounds for linear PTS, including Figure~\ref{fig:running-example-linear-loop}. We note that our work for bounding (expected) values of random variables is not specific to linear loops, or to loop analysis in general, and can further be extended with  %There is a lot of room for improvement, for instance adding 
    new derivation rules.
    %, or developing ways to remove/avoid generating redundant bounds.
    {Similarly, Theorem~\ref{thm: ost} is not bound to the application to linear programs. It may be possible to derive analogous results to Theorem~\ref{thm:moment-monomials} for different classes of programs.}

    \paragraph{{\bf Multiplication of monomial bounds.}} For bounds of monomials, we  use  algebraic properties based on the signs of existing bounds. In total, there are 5 rules for deriving new upper bounds, and 7 rules for lower bounds, some of which are shown below.

        \begin{mathpar}
      \inferrule*[right=(square-positive)]
        {~}
        {\vec{x}^{2\alpha} \geq 0}
    \and
      \inferrule*[right=(mul-ub-1)]
        {a\leq \vec{x}^\alpha\leq b\\ c\leq \vec{x}^\beta \leq d\\a\geq 0\\c\geq 0}
        {\vec{x}^{\alpha+\beta} \leq bd}
    %\and
    %  \inferrule*[right=(mul-lb-1)]
    %    {\vec{x}^\alpha\leq b\\ \vec{x}^\beta \leq d\\ b\leq 0\\d\leq 0}
    %    {bd\leq \vec{x}^{\alpha+\beta}}
    \end{mathpar}

    \paragraph{{\bf Bounds for moments.}} Bounds for moments can either be derived directly from a bound of the value of a monomial, or through Jensen's inequality~\cite{klenke}. Additionally, we provide three multiplication rules. Those rules exploit, in addition to bounds of moments, the signs of the random variables in question. Appendix~\ref{appendix:derivation-rules} shows the soundness of our derivation rules. 

        \begin{mathpar}
      \inferrule*[right=(moment-lb)]
        {a\leq \vec{x}^{\alpha}}
        {a\leq \mathbb{E}(\vec{x}^{\alpha})}
    \and
      \inferrule*[right=(moment-ub)]
        {\vec{x}^{\alpha}\leq a}
        {\mathbb{E}(\vec{x}^{\alpha})\leq a}
    \and
        \inferrule*[right=(jensen-lb-1)]
        {a\leq \mathbb{E}(\vec{x}^\alpha)\\ 0 \leq a}
        {a^2\leq \mathbb{E}(\vec{x}^{2\alpha})}
        \and
        %\inferrule*[right=(jensen-lb-2)]
        %{\mathbb{E}(\vec{x}^\alpha)\leq a\\ a\leq 0}
        %{a^2\leq \mathbb{E}(\vec{x}^{2\alpha})}
        %\and
        \inferrule*[right=(rv-mul-1)]
        {a\leq \vec{x}^\alpha \\ a\leq 0 \\ 0\leq \vec{x}^\beta\\ \mathbb{E}(\vec{x}^\beta) \leq b}
        {ab\leq \mathbb{E}(\vec{x}^{\alpha+\beta})}
% \and
%         \inferrule*[right=(rv-mul-2)]
%         {\vec{x}^\alpha\leq a \\ a\geq 0 \\ 0\leq \vec{x}^\beta\\ \mathbb{E}(\vec{x}^\beta) \leq b\\ 0\leq b}
%         {\mathbb{E}(\vec{x}^{\alpha+\beta})\leq ab}
%         \and
%         \inferrule*[right=(rv-mul-3)]
%         {a\leq \vec{x}^\alpha \\ a\geq 0 \\ 0\leq \vec{x}^\beta\\ b\leq\mathbb{E}(\vec{x}^\beta)\\ 0\leq b}
%         {ab\leq\mathbb{E}(\vec{x}^{\alpha+\beta})}
    
    \end{mathpar}

    \paragraph{{\bf Derivation through martingales.}}  Derivation rules are also created based on the synthesised polynomial martingales which satisfy the OST via the {\dui} precondition. 
    For such a martingale $m=\sum_{\alpha\in\mathcal{A}} a_\alpha \vec{x}^\alpha$, we set $\mathcal{A}^+:=\{\alpha\in\mathcal{A}\mid a_\alpha~>~0~\}$ and $\mathcal{A}^-:=\{\alpha\in\mathcal{A}\mid a_\alpha<0\}$. 
    We denote by $m_0:=\sum_{\alpha\in\mathcal{A}} a_\alpha \vec{x}_0^\alpha$ the (potentially symbolic) initial value of the polynomial. 
    As $m$ satisfies the OST, $\mathbb{E}(m_T) =m_0$. From upper/lower bounds of the expected value of all but one monomial $\vec{x}_T^\alpha$ of the polynomial, we derive a new bound for $\mathbb{E}(\vec{x}_T^\alpha)$. With $u_\alpha$ and $l_\alpha$ being existing upper and lower bounds on monomials, we use the following rules (ub-martingale) and (lb-martingale), assuming w.l.o.g. $\alpha\in\mathcal{A}^+$ (otherwise we multiply the martingale by $-1$).

    \begin{mathpar}
        \inferrule*[right=(ub-martingale)]
        {\alpha\in\mathcal{A}^+\\ \forall_{\beta\in\mathcal{A}^+\setminus\{\alpha\}}: l_\beta\leq \mathbb{E}(\vec{x}^\beta)\\\forall_{\beta\in\mathcal{A}^-}:  \mathbb{E}(\vec{x}^\beta)\leq u_\beta}
    {\textstyle \mathbb{E}(\vec{x}^\alpha) \leq \frac{1}{a_\alpha} \left(m_0 - \sum_{\beta\in\mathcal{A}^+\setminus\{\alpha\}} l_\beta a_\beta+ \sum_{\beta\in\mathcal{A}^-} u_\beta a_\beta\right)}
    \and
    \inferrule*[right=(lb-martingale)]
        {\alpha\in\mathcal{A}^+\\ \forall_{\beta\in\mathcal{A}^+\setminus\{\alpha\}}: \mathbb{E}(\vec{x}^\beta)\leq u_\beta\\\forall_{\beta\in\mathcal{A}^-}: l_\beta\leq \mathbb{E}(\vec{x}^\beta)}
    {\textstyle \frac{1}{a_\alpha} \left(m_0 - \sum_{\beta\in\mathcal{A}^+\setminus\{\alpha\}} u_\beta a_\beta+ \sum_{\beta\in\mathcal{A}^-} l_\beta a_\beta\right)\leq \mathbb{E}(\vec{x}^\alpha)}
    \end{mathpar}

    \begin{example}
        \label{example:bound-derivation}
         In Figure~\ref{fig:running-example-linear-loop}, by the negated loop guard and by knowing the support of 
        $\text{Uniform}(-1,0)$, we have $-1 \leq x_T \leq 0$. %$x_T\leq 0$ and $x_T\geq -1$. 
        So, by using the (moment-lb) and (moment-ub) rules, %it follows that 
        $-1\leq\mathbb{E}(x_T)\leq0$. Also by standard interval propagation $0\leq x_T^2\leq 1$, and thus, $0\leq \mathbb{E}(x_T^2)\leq 1$. 
        Because the loop is entered at least once, we know that $k_T\geq1$. 
        Moreover, by using the (ub-martingale) and (lb-martingale) rules for the martingale in Eq.~\eqref{eq:martingale-k-x}, we derive that $2x_0\leq\mathbb{E}(k_T)\leq 2 x_0 + 2$. 
        Similarly, by using Eq.~\eqref{eq:martingale-x-y} and the same rules, 
        we obtain that 
        $x_0+y_0\leq \mathbb{E}(y_T)\leq x_0+y_0+1$. % {holds}.

        Furthermore, 
        using the (rv-mul-1) rule we derive the bound $-2x_0-2 \leq\mathbb{E}(k_Tx_T)$, since we know that $-1\leq x_T$ and $1\leq k_T$ as well as $\mathbb{E}(k_T)\leq 2x_0+2$. This is useful, as it enables us to derive the upper bound $\mathbb{E}(k_T^2)\leq 4x_0^2+\frac{26}{3}x_0 + \frac{26}{3}$ through the (ub-martingale) rule and the martingale in Eq.~\eqref{eq:martingale-for-deriving-k2}.
    \end{example}

    \paragraph{{\bf Cauchy-Schwarz.}} Since the polynomial martingales are not capable of tracking all relations between variables, purely linear propagation of bounds via the rules (ub-martingale) and (lb-martingale) is usually insufficient. In some cases, instead of obtaining a bound on $\mathbb{E}(\vec{x}^\alpha)$ directly, it is however possible to obtain one for $\mathbb{E}((a\vec{x}^{\beta}+b\vec{x}^{\gamma})^2)$, such that $\alpha=\gamma+\beta$. In that case, we can use the rules (cs-lb) and (cs-ub), to derive bounds for $\mathbb{E}(\vec{x}^\alpha)$. Those rules are direct consequences of the Cauchy-Schwarz inequality~\cite{durrett}.

    \begin{mathpar}
        \inferrule*[right=(cs-lb)]
        {\mathbb{E}((a\vec{x}^\alpha+b\vec{x}^\beta)^2) \leq g\\ \mathbb{E}(\vec{x}^{2\alpha})\leq h}
        {-\frac{\sqrt{gh}}{|b|}-\frac{h}{2}\left(\left|\frac{a}{b}\right|+\frac{a}{b}\right)\leq \mathbb{E}(\vec{x}^{\alpha+\beta})}
        \and
        \hspace{-0.1em}
        \inferrule*[right=(cs-ub)]
        {\mathbb{E}((a\vec{x}^\alpha+b\vec{x}^\beta)^2) \leq g\\ \mathbb{E}(\vec{x}^{2\alpha})\leq h}
        {\mathbb{E}(\vec{x}^{\alpha+\beta})\leq \frac{\sqrt{gh}}{|b|}+\frac{h}{2}\left(\left|\frac{a}{b}\right|-\frac{a}{b}\right)}
    \end{mathpar}

    To exploit this, we perform square completion by pattern matching, to find expressions of the form $\mathbb{E}((a\vec{x}^\alpha+b\vec{x}^\beta)^2)$ that yield martingales with fewer summands. The other rules are then also instantiated with those expressions.

    \begin{example}
        \label{example:square-extraction-cauchy-schwarz}
        We aim to derive an upper bound for the expression $\mathbb{E}(x_Ty_T)$. Such a bound cannot be directly derived from the martingale in Eq.~\eqref{eq:martingale-for-square-extraction}, since we are lacking an appropriate bound for $\mathbb{E}(y_T^2)$. We can, however, rewrite this martingale as $14\mathbb E(x_T) + 3\mathbb{E}((x_T+y_T)^2)=14x_0 + 3x_0^2 + 6 x_0y_0 + 3y_0^2$. Since we have a lower bound for $\mathbb{E}(x_T)$ (Example~\ref{example:bound-derivation}), we get that $\mathbb{E}((x_T+y_T)^2)\leq x_0^2+2x_0y_0+ y_0^2 +\frac{14}{3}x_0+ \frac{14}{3}$. Then we can use the (cs-ub) rule, since we already obtained a bound for $\mathbb{E}(x_T^2)$ in Example~\ref{example:bound-derivation}:
        $\mathbb{E}(x_T y_T)\leq \sqrt{x_0^2+2x_0y_0 + y_0^2 +\frac{14}{3}x_0+ \frac{14}{3}}$.
    \end{example}

    \paragraph{{\bf Limitations of automatic bound generation.}} The martingale rule is applied for all possible combinations of bounds. To avoid such computationally expensive tasks, we employ aggressive subsumption checking. To facilitate this, we restrict square roots to only contain one single initial variable or numeric values, and only compare the highest degree monomials. After applying the rules based on Cauchy-Schwarz, we use the sub-additivity of square roots.
    The equation system defined in Eq.~\eqref{eq:linear-system-of-equations} is usually under-determined, and enumeration of all rules is too costly. 
    We therefore select a number of sparsest solutions of the equation system (see Appendix~\ref{appendix:limitations-of-implementation}).

    \section{Examples of Bound Derivation}
    \label{section:experimental-results}

    We show that our approach is suitable for computing symbolic bounds for the expected value of monomials after loop termination. To this end, we implement the automatic bound derivation in \texttt{Polar} and analyse Figures~\ref{fig:running-example-linear-loop}~\&~\ref{fig:example-unbounded-update-for-lg}, for both of which the computation of tight symbolic bounds for moments after termination is out of reach for existing tools.
 
    In Table~\ref{tab:experimental-results} we list symbolic bounds computed for different inputs. Column~1 specifies which moment of the stopping time of the analysed program must be finite in order for the derivation to be sound. For the analysis to be sound, we must verify that $\mathbb{E}(T^2)$ is finite for the bounds derived for Figure~\ref{fig:running-example-linear-loop}, and $\mathbb{E}(T^3)<\infty$ for the bounds for Figure~\ref{fig:example-unbounded-update-for-lg}, since the stopping time depends on variables with unbounded support (see Theorem~\ref{thm:moment-monomials}). We currently require those conditions to be verified externally. The condition $\mathbb{E}(T^2)<\infty$ for Figure~\ref{fig:running-example-linear-loop} can be verified using state-of-the-art tools~\cite{wang2021}, since the variable in the loop guard is not affected by unbounded support updates. To the best of our knowledge, this is however not the case for Figure~\ref{fig:example-unbounded-update-for-lg}, as existing methods  struggle with unbounded updates. 
    We show,  through an application of tail bounds in Appendix~\ref{appendix:stopping-times-of-examples}, that the required bounds for the moments of the stopping times hold. 
    
    While we still require manual intervention, our implementation demonstrates the benefits of our two-step approach, as the tail bounds in the manual proof are powerful for asymptotic analysis but usually do not yield useful symbolic bounds. Those symbolic bounds can then be computed as a second step by our approach\footnote{It is always possible to introduce a program variable that simply updates by one in every step of the program, effectively storing the runtime. Hence, our approach is equally capable of deriving upper and lower bounds on moments of the runtime. This is the variable $k$ in Figures~\ref{fig:running-example-linear-loop}~\&~\ref{fig:example-unbounded-update-for-lg}.}.

    \begin{wrapfigure}{r}{0.30\textwidth}%{0.42\textwidth}
		\begin{algorithmic}
			\STATE{$z := z_0$}
			\STATE{$x := x_0$}
            \STATE{$k:=0$}
			\WHILE{$x \geq 0$} 
			\STATE{$k:= k+1$}
			\STATE{$c \sim \mathrm{Bernoulli}(\frac{7}{10})$}
            \IF{$c=1$}
			    \STATE {$z\sim\mathrm{Normal}(-1,1)$}
			         \STATE {$x:=x+z$}
            \ELSE
			      \STATE {$z\sim\mathrm{Normal}(1,1)$}
                \STATE{$x:=x+z$}
            \ENDIF
			\ENDWHILE	
		\end{algorithmic}
        \caption{Linear loop with unbounded update affecting the loop guard.}
        \label{fig:example-unbounded-update-for-lg}
    \end{wrapfigure}
    
    Column~1 of Table~\ref{tab:experimental-results} states the set of initially known inequalities. 
    The bound for $\mathbb{E}(x_T)$ in Figure~\ref{fig:example-unbounded-update-for-lg} is obtained by reasoning about the last step: $\mathbb{E}(x_T) = \mathbb{E}(x+Y\mid x>0\land x+Y <0)$, where $Y$ is the distribution of the last update. This update $Y$ is lower bounded in expectation by $Y'=\text{Normal}(-1,1)$, hence $\mathbb{E}(x_T) \geq\mathbb{E}(Y'\mid Y'<0)$, which is equal to the mean of a truncated random variable, which for these parameters is lower bounded by $-1.3$. By analogous reasoning using a truncated random variable, we get that $\mathbb{E}(x_T^2)\leq \mathbb{E}((Y')^2\mid Y'<0)\leq2.3$. {We additionally assume that the variable $x_0$ is positive in both examples; this is equivalent to assuming that the loop is executed at least once.
    
    Column~3 of Table~\ref{tab:experimental-results} shows  that, in many cases, the obtained bounds are asymptotically sharp, in a sense that the leading terms of upper and lower bounds are equal. This holds also for the bounds for $\mathbb{E}(z_T)$ in Figure~\ref{fig:running-example-linear-loop}, although the increments of $z$ are not bounded from below, and are affected by the value of $y$. To the best of our knowledge, no other approach for cost analysis or bound computation is able to handle those properties simultaneously.
    }
\renewcommand{\arraystretch}{1.6}

\newcommand{\FigOneLBk}{\ensuremath{2x_0}}
\newcommand{\FigOneUBk}{\ensuremath{2x_0+2}}
\newcommand{\FigOneLBy}{\ensuremath{x_0+y_0}}
\newcommand{\FigOneUBy}{\ensuremath{x_0+y_0+1}}

\newcommand{\FigOneLBz}{\ensuremath{x_0^2 + 2x_0y_0 - 5x_0 + z_0 - 2|y_0|-\sqrt{\frac{56}{3}}\sqrt{x_0} - \sqrt{\frac{56}{3}} - 4}}
\newcommand{\FigOneUBz}{\ensuremath{x_0^2 + 2x_0y_0 - x_0 + z_0 + 2|y_0|+\sqrt{\frac{56}{3}}\sqrt{x_0} + \sqrt{\frac{56}{3}}+2}}

\newcommand{\FigOneLBkx}{\ensuremath{-2x_0-2}}
\newcommand{\FigOneUBkx}{\ensuremath{0}}
\newcommand{\FigOneLBky}{\ensuremath{2x_0^2 + 2x_0y_0 - \frac{4}{3}x_0 - 2|y_0|-\sqrt{\frac{56}{3}}\sqrt{x_0} - \sqrt{\frac{56}{3}} - 2}}
\newcommand{\FigOneUBky}{\ensuremath{2x_0^2 + 2x_0y_0 + \frac{14}{3}x_0 + 2|y_0| +\sqrt{\frac{56}{3}}\sqrt{x_0}+ \sqrt{\frac{56}{3}} + \frac{14}{3}}}

\newcommand{\FigOneLBxy}{\ensuremath{- x_0 -|y_0| -\sqrt{\frac{14}{3}}\sqrt{x_0}- \sqrt{\frac{14}{3}} - 1}}
\newcommand{\FigOneUBxy}{\ensuremath{x_0 + |y_0|+\sqrt{\frac{14}{3}}\sqrt{x_0} + \sqrt{\frac{14}{3}}}}
\newcommand{\FigOneLBkk}{\ensuremath{4x_0^2 + \frac{2}{3}x_0 - 4}}
\newcommand{\FigOneUBkk}{\ensuremath{4x_0^2 + \frac{26}{3}x_0 + \frac{26}{3}}}
\newcommand{\FigOneLByy}{\ensuremath{x_0^2 + 2x_0y_0 + y_0^2+ \frac{8}{3}x_0  - 2|y_0| -\sqrt{\frac{56}{3}}\sqrt{x_0}- \sqrt{\frac{56}{3}}-1}}
\newcommand{\FigOneUByy}{\ensuremath{x_0^2 + 2x_0y_0+ y_0^2 + \frac{20}{3}x_0  + 2|y_0| +\sqrt{\frac{56}{3}}\sqrt{x_0}+ \sqrt{\frac{56}{3}} + \frac{20}{3}}}

\newcommand{\FigTwoLBk}{\ensuremath{\frac{5}{2}x_0}}
\newcommand{\FigTwoUBk}{\ensuremath{\frac{5}{2}x_0 + \frac{13}{4}}}
    
\newcommand{\FigTwoLBkk}{\ensuremath{\frac{25}{4}x_0^2 + \frac{115}{4}x_0 - \frac{115}{8}}}
\newcommand{\FigTwoUBkk}{\ensuremath{\frac{25}{4}x_0^2 + \sqrt{\frac{2875}{8}}x_0 + \frac{115}{4}x_0 +\sqrt{\frac{13225}{8}}\sqrt{x_0}+ \sqrt{\frac{34385}{16}} + \frac{529}{8}}}

\newcommand{\FigTwoLBkx}{\ensuremath{- \sqrt{\frac{115}{8}}x_0-\sqrt{\frac{529}{8}}\sqrt{x_0} - \sqrt{\frac{6877}{80}} - \frac{23}{4}}}
\newcommand{\FigTwoUBkx}{\ensuremath{0}}

  \begin{table}
    \centering
    \begin{tabular}{l|l|l|l}
        \shortstack[c]{[program]\\assumptions}& moment & \multicolumn{2}{c}{bound}\\ \hline
        \multirow{12}{*}{\shortstack[l]{[Fig.~\ref{fig:running-example-linear-loop}]\\$\mathbb{E}(T^2)<\infty$\\
        $x_T\leq0$\\$x_T\geq -1$\\$x_0>0$}}
          &\multirow{1}{*}{$\mathbb{E}(k_T)$} & $\geq\FigOneLBk$ & $\leq \FigOneUBk$\\\cline{2-4}

          &\multirow{1}{*}{$\mathbb{E}(y_T)$} & $\geq\FigOneLBy$ & $\leq \FigOneUBy$\\\cline{2-4}

          &\multirow{2}{*}{$\mathbb{E}(z_T)$} & \multicolumn{2}{l}{$\geq\FigOneLBz$} \\\cline{3-4}
          & & \multicolumn{2}{l}{$\leq \FigOneUBz$}\\\cline{2-4}

          &\multirow{1}{*}{$\mathbb{E}(k_Tx_T)$} & $\geq\FigOneLBkx$ & $\leq \FigOneUBkx$\\\cline{2-4}

          &\multirow{2}{*}{$\mathbb{E}(k_Ty_T)$} & \multicolumn{2}{l}{$\geq\FigOneLBky$} \\\cline{3-4}
          & & \multicolumn{2}{l}{$\leq \FigOneUBky$}\\\cline{2-4}

          &\multirow{2}{*}{$\mathbb{E}(x_Ty_T)$} & \multicolumn{2}{l}{$\geq\FigOneLBxy$} \\\cline{3-4}
          & & \multicolumn{2}{l}{$\leq \FigOneUBxy$}\\\cline{2-4}

          &\multirow{1}{*}{$\mathbb{E}(k_T^2)$} & $\geq\FigOneLBkk$ \;\;\;\;\;\;\;\;\;\;\;\;\;\;\;\;\;& $\leq \FigOneUBkk$\\\cline{2-4}

          &\multirow{2}{*}{$\mathbb{E}(y_T^2)$} & \multicolumn{2}{l}{$\geq\FigOneLByy$} \\\cline{3-4}
          & & \multicolumn{2}{l}{$\leq \FigOneUByy$}\\\hline
    \multirow{5}{*}{\shortstack[l]{
        [Fig.~\ref{fig:example-unbounded-update-for-lg}]\\
        $\mathbb{E}(T^3)<\infty$\\ $\mathbb{E}(x_T)\geq-1.3$\\$\mathbb{E}(x_T^2)\leq 2.3$\\$x_0>0$
        }}
        &\multirow{1}{*}{$\mathbb{E}(k_T)$} & $\geq\FigTwoLBk$ & $\leq \FigTwoUBk$\\\cline{2-4}
        &\multirow{2}{*}{$\mathbb{E}(k_T^2)$} & \multicolumn{2}{l}{ $\geq\FigTwoLBkk$} \\\cline{3-4}
        & & \multicolumn{2}{l}{$\leq \FigTwoUBkk$}\\\cline{2-4}
        &\multirow{2}{*}{$\mathbb{E}(k_Tx_T)$} & \multicolumn{2}{l}{$\geq\FigTwoLBkx$} \\\cline{3-4}
        & & \multicolumn{2}{l}{$\leq \FigTwoUBkx$}\\\hline
        
    \end{tabular}
    \caption{Symbolic bounds on moments of PTS.}
    \label{tab:experimental-results}
\end{table}
    
\renewcommand{\arraystretch}{1}

	\section{Related Work} \label{sec: relatedwork}

	Recently, there have been significant advances in two related problems: \emph{synthesis of probabilistic invariants} and \emph{reasoning about (positive) almost sure termination}. A deep link between these two problems is evident in the preceding, especially in our setting of linear PTS from Section~\ref{sec: linear}.
	A telling perspective is provided by the Optional Stopping Theorem; the OST % (which underpins a recent approaches to synthesising invariants), %\ref{thm: OST}, 
	typically requires the runtime, the martingale expression or both to be bounded in some sense. 
	Other examples of a mixed precondition for the OST are given in~\cite{costanalysis,wang2021,piecewise_analysis_prob2026,chatterjee2024quantitative}; for an overview, see Appendix~\ref{sec: appendix ost}. 
    The precondition of~\cite{wang2021} uses martingales $\{Y_n\}_{n \in \mathbb{N}}$ that are polynomially bounded in $n$ and expects that  sufficiently high moments of the runtime are finite. 
    Hence, the setting of~\cite{wang2021} only supports  sampling from bounded distributions. Similarly, in~\cite{costanalysis} such a polynomial bound is required for the step-wise difference of a martingale, together with exponential tail decay of the stopping time. In~\cite{piecewise_analysis_prob2026} a similar bound for the stopping time is required, but the step-wise difference can grow exponentially. 
    Nonetheless, all of these (almost sure) step-wise difference bounds are incompatible with sampling from distributions with unbounded support. 
    Moreover, in~\cite{chatterjee2024quantitative} such a step-wise difference bound is not required, but the martingales are required to be non-negative, which restricts their applicability for unbounded support distributions. 
    Unlike this, our result from Theorem~\ref{thm:moment-monomials} handles unbounded distributions and martingale invariants without a sign constraint for a (restricted) class of PTS programs.

    \smallskip
    
	This link between probabilistic invariants and runtime bounds is further investigated  in \cite{DBLP:journals/pacmpl/HarkKGK20}. Here, it is shown that for computing lower bounds on the runtime of probabilistic programs,  it is necessary to satisfy the OST.
	As a natural counterpart to upper bounds, the synthesis of lower bounds of the expected runtime has been studied in \cite{DBLP:journals/toplas/FrohnNBG20,DBLP:conf/vmcai/FuC19,feng2023lower}.

	Conversely, upper bounds on the expected runtime can be improved via probabilistic invariants. 
    Upper bounds of the expected runtime can be found by finding suitable non-negative supermartingales~\cite{e,DBLP:conf/tacas/KuraUH19,d,a,c,wang2021}. 
	Once the OST is satisfied,  our probabilistic invariants can take both negative and positive values, extending thus~\cite{chatterjee2024quantitative} where  moments of polynomials of program costs are derived provided a fixed lower bound exists. 
    
	Higher-order moments of program variables are inferred in~\cite{forsyte,DBLP:conf/tacas/BartocciKS20,costanalysis}. 
    Importantly, in~\cite{feng2023lower} the authors provide lower bounds for expected program quantities, even if the program in question is not almost surely terminating. This is achieved by reducing the question to a well-chosen almost surely terminating program. 
	The general approach of checking preconditions of the OST for  polynomial invariants of higher-order is further studied in~\cite{wp1}, providing  sufficient conditions that have been successfully implemented in  \cite{DBLP:conf/atva/FengZJZX17,DBLP:conf/cav/ChenHWZ15,DBLP:conf/qest/GretzKM13}. 

    \section{Conclusion}
    We study the invariant synthesis problem for probabilistic transition systems, by leveraging the Optional Stopping Theorem~e.g.~\cite{cohen}, for which we present a general sufficient precondition, coined as \dui. Our \dui yields a criterion under which a given polynomial expression qualifies as an invariant. We use this criterion, to automatically synthesise polynomial probabilistic invariants for a certain class of linear loops, and use the invariants to automatically derive bounds for higher- and mixed-moments of program variables {and the runtime} after loop termination. Notably, our approach handles distributions with unbounded support, enabling thus the analysis of loops that cannot be handled by other techniques. Future work includes handling {nonlinear} loops for which our \dui precondition can be used or adjusted, thus further improving the automatic derivation of moment bounds. 
    {These methods are especially promising regarding nested loops. Such an approach would necessarily involve unbounded updates, as both termination time and variable updates of inner loops can scarcely be expected to be bounded by constants. 
    The flexibility of the \dui precondition could provide tractable sufficient conditions for the applicability of the OST under compositional reasoning.}
    
   \paragraph{Acknowledgments.} 

    A.S.~is supported by the Cantab Capital Institute for the Mathematics of Information and by the Academy of Finland Centre of Excellence Programme grant number
346315 ''Finnish centre of excellence in Randomness and STructures (FiRST)``.
   
   This research was partially funded by the Austrian Science Fund (FWF) [Grant ID 10.55776/DOC1345324] and the Vienna Science and Technology Fund (WWTF), State of Lower Austria [Grant ID 10.47379/ICT25017].

	\newpage

	\bibliographystyle{splncs04}
	\bibliography{references}
	
	%%%%%%%%%%%%%%%%%%%%%%%%%%%%%%%%%%%%%%%%%%%%%%%%%%%%%%%%%%%%%%%%%%%%%%%%%%%%%
	% Future work and open problems %
	%%%%%%%%%%%%%%%%%%%%%%%%%%%%%%%%%%%%%%%%%%%%%%%%%%%%%%%%%%%%%%%%%%%%%%%%%%%%%%
	
	%	\section{Future work and open problems} \label{sec: future}
	%	
	%	There are several interesting directions of further research.
	%	Firstly, one question is whether this method can be generalised for probabilistic transition system with labels. %as presented in \cite{a, d}.
	%	Secondly, the key part of this method is to find invariants in the form of martingale expression 
	%	%is to verify that such expressions 
	%	that fulfil the second condition Optional Stopping Theorem.
	%	It is interesting to study whether other conditions of this theorem can be automatically verified.
	%	%Thirdly, this method is still very basic. It would be worthwhile to see whether it can be improved. In particular, 		
	%	Lastly, a very important direction of the further research is to take the numerical aspects
	%	% of sum-of-squares optimisation 
	%	into account. The paper \cite{numerics} presents many of the common issues of using sos optimisation in static analysis.
	
	%Refinement, computational effort.
	%scalability
	%transition systems with labels, polynomial stochastic systems
	
	\clearpage
	
	%%%%%%%%%%%%%%%%%%%%%%%%%%%%%%%%%%%%%%%%%%%%%%%%%%%%%%%%%%%%%%%%%%%%%%%%%%%%%
	% Appendices %
	%%%%%%%%%%%%%%%%%%%%%%%%%%%%%%%%%%%%%%%%%%%%%%%%%%%%%%%%%%%%%%%%%%%%%%%%%%%%%%
	
	\appendix
	
    \section{A crash course in martingale theory} 
    \label{section: appendix martingales}

    In this section we review the definitions of martingales and related results by virtue of some concrete examples. 
    We hope this is helpful for readers with less familiarity with these topics. 
    It is not necessary for the main sections of this paper and it will not be referenced outside of this section. 
    Nonetheless, it should hopefully give some insight on the intuition behind some of the more abstract concepts of probability theory.

    \paragraph{$\sigma$-algebra} 
    Let $\Omega$ be a set. Then a subset $\calF \subseteq \mathcal{P}(\Omega)$\footnote{Here $\mathcal{P}$ denotes the power set.} 
    is a $\sigma$-algebra, if 
    \begin{enumerate}
        \item $\Omega \in \calF$,
        \item for all $A \in \calF$ the complement $(\Omega \setminus A) \in \calF$,
        \item for all $A_1, A_2, \cdots \in \calF$ the countable union $\bigcup_{k = 1}^{\infty} A_k \in \calF$.
    \end{enumerate}
    If $B \subseteq \mathcal{P}(\Omega)$ is arbitrary, then we write $\sigma(B)$ for the smallest $\sigma$-algebra containing $B$.

    \paragraph{Probability measure}
    Let $\Omega$ be a set and let $\calF$ be a $\sigma$-algebra over $\Omega$. 
    Then a probability measure $\mathbb{P}: \calF \rightarrow [0, 1]$ is a function such that 
    $\mathbb{P}(\emptyset) = 0$, $\mathbb{P}(\Omega) = 1$ and 
    \begin{align*}
        \mathbb{P} \left( \bigcup_{k = 1}^{\infty} A_k \right) 
        = \sum_{k = 1}^{\infty}  \mathbb{P} \left( A_k \right) 
        \qquad \qquad 
        \mathrm{for \ all \ disjoint \ } A_1, A_2, \cdots \in \calF.
    \end{align*}
    
    \paragraph{Probability space}
    For this example we fix a probability space $(\Omega, \calF, \mathbb{P})$ as follows. 
    The sample space is $\Omega = [0, 1)$ and as $\sigma$-algebra $\calF$  we choose the Borel $\sigma$-algebra on $[0, 1)$ denoted as $\mathbb{B}([0, 1))$, which is the $\sigma$-algebra generated by open subsets of $[0, 1)$. 
    Lastly, we let the probability measure~$\mathbb{P}$ be the uniform measure on $[0, 1)$:
    \begin{align*}
        \mathbb{P} \left( A \right) = \int_A 1 \ dx  
        \qquad \qquad \mathrm{for \ all \ (Borel) \ measurable \ } A \in \mathbb{B}([0, 1)).
    \end{align*}

    \paragraph{Random variable}
    Moreover, we fix a random variable $Z: \Omega \rightarrow [0, 1]$ that is uniformly distributed.
    Specifically, we set $Z: [0, 1) \rightarrow [0, 1]$ as the identity, i.e.~$Z(\omega) = \omega$ for all $\omega \in [0, 1)$.
    This implies that for all (Borel) measurable $A \subseteq [0, 1)$,
    \begin{align*}
        \mathbb{P} \left( \left\lbrace \omega \in \Omega : Z(\omega) \in A \right\rbrace \right) 
        %= \mathbb{P} \left( \left\lbrace \omega \in [0, 1) : Z(\omega) \in A \right\rbrace \right)  
        = \int_A 1 \ dx.
    \end{align*}

    \paragraph{Filtration}
    Now we equip this probability space $(\Omega, \calF, \mathbb{P})$ with a filtration $(\calF_n)_{n \in \mathbb{N}}$ generated by dyadic intervals $D_n := \left\lbrace [k \cdot 2^{-n}, (k+1) \cdot 2^{-n}) : k \in \{0, 1, \cdots, 2^n - 1 \} \right\rbrace$:
    \begin{itemize}
        \item $\calF_0 = \sigma \left(D_0 \right) = \left\lbrace \emptyset, [0, 1) \right\rbrace$, 
        \item $\calF_1 = \sigma \left(D_1 \right) = \left\lbrace \emptyset, [0, 1/2), [1/2, 1), [0, 1) \right\rbrace$, 
        \item $\calF_n = \sigma \left(D_n \right) = \left\lbrace \bigcup_{k \in I} [k \cdot 2^{-n}, (k+1) \cdot 2^{-n}) : I \subseteq  \{0, 1, \cdots, 2^n - 1 \} \right\rbrace$ for all $n \in \mathbb{N}$.
    \end{itemize}
    In particular, we have  $\calF_0 \subseteq \calF_n \subseteq \calF_{n + 1} \subseteq \calF$ for all $n \in \mathbb{N}$, thus 
    $([0, 1), \calF, \{\calF_n\}_{n \in \mathbb{N}}, \mathbb{P})$ is indeed a filtered probability space. 

    \paragraph{Conditional expectation} 
    We want to compute the conditional expectation $Z_n := \mathbb{E} \left[ Z \mid \calF_n \right]$ for all $n \in \mathbb{N}$. 
    The {\em conditional expectation} is characterised by two defining properties: 
    \begin{enumerate}
        \item The random variable $Z_n = \mathbb{E} \left[ Z \mid \calF_n \right]$ is $\calF_n$-measurable. 
        \item We have 
        %For all $A \in \calF_n$,
        \begin{align}
            \label{eq: second characteristic conditional expecatation}
            \int_A Z_n(x) \ dx = \mathbb{E} \left[1_A \cdot Z_n \right] = \mathbb{E} \left[1_A \cdot Z \right] = \int_A x \ dx 
        \qquad \qquad \mathrm{for \ all \ } A \in \calF_n.
        \end{align}
    \end{enumerate}
%    In particular, by the definition of $Z$ and $\mathbb{P}$ we have 
%    \begin{align*}
%        \mathbb{E} \left[1_A \cdot Z \right] = \int_A Z(x) \ dx = \int_A x \ dx.
%    \end{align*}
%    Therefore, the conditional expectation $Z_n = \mathbb{E} \left[ Z \mid \calF_n \right]$ is defined s.t.~$Z_n$ is $\calF_n$-measurable and 
%    \begin{align}
%        \label{eq: definition Zn}
%        \int_A Z_n(x) \ dx = \int_A x \ dx 
%        \qquad \qquad \mathrm{for \ all \ } A \in \calF_n.
%    \end{align}
%   \noindent
    Hence, let's compute $Z_0$. First $Z_0$ is $\calF_0$-measurable, i.e.~for all open $A \subseteq [0, 1]$ we have 
    \begin{align*}
        Z_n^{-1}(A) \in \calF_0  = \left\lbrace \emptyset, [0, 1) \right\rbrace.
    \end{align*}
    This is only possible, if the function $Z_0 : [0, 1) \rightarrow [0, 1]$ is constant. 
    So, there exists a $c_0 \in [0, 1]$ such that $Z_0(\omega) = c_0$ for all $\omega \in \Omega = [0, 1)$. 
    Moreover, by \eqref{eq: second characteristic conditional expecatation} and by choosing $A = [0, 1) \in \calF_0$, 
    \begin{align*}
        c_0 = \int_{[0, 1)} Z_n(x) \ dx = \int_{[0, 1)} x \ dx = \frac{1}{2} = \mathbb{E} \left[ Z \right].
    \end{align*}
    In conclusion, $Z_0: [0, 1) \rightarrow [0, 1]$ is defined by $Z_0(\omega) = 1/2$ for all $\omega \in [0, 1)$.

    \begin{remark}
    \label{remark: first intuition conditional expectation}
        This example works more generally.  
        Let $X: \Omega \rightarrow \mathbb{R}$ be some random variable defined on some probability space $(\Omega, \calF, \mathbb{P})$. 
        Because every $\sigma$-algebra contains the empty set and the whole sample space, the trivial $\sigma$-algebra $\calF_0 := \{ \emptyset, \Omega \}$ is a subset of $\calF$. 
        In particular, by the same argument the random variable $\mathbb{E} \left[ X \mid \calF_0 \right]$ is constant 
        and hence for (almost) all $\omega \in \Omega$, 
        \begin{align*}
            \mathbb{E} \left[ X \mid \calF_0 \right](\omega) 
            = \mathbb{E} \left[ \mathbb{E} \left[ X \mid \calF_0 \right] \right] 
            = \mathbb{E} \left[ 1_{\Omega} \cdot \mathbb{E} \left[ X \mid \calF_0 \right] \right] 
            = \mathbb{E} \left[ 1_{\Omega} \cdot X  \right] 
            %= \int_{\Omega} \mathbb{E}[X \mid \calF_0] \ d \mathbb{P} 
            %= \int_{\Omega} X \ d \mathbb{P}  
            = \mathbb{E} \left[X\right].
        \end{align*}
    The intuition here is that a sub-$\sigma$-algebra $\calG \subseteq \calF$ represents information that can be known. 
    Then the full $\sigma$-algebra $\calF$ corresponds to having full information on $X: \Omega \rightarrow \mathbb{R}$, whereas a sub-$\sigma$-algebra $\calG \subseteq \calF$ represents a restriction of information (expressed by the $\calG$-measurability). 
    Consequently, the random variable $\mathbb{E} \left[ X \mid \calG \right]$ is the projection/best approximation of $X: \Omega \rightarrow \mathbb{R}$ under this restriction (expressed by the second property of the conditional expectation~e.g.~\eqref{eq: second characteristic conditional expecatation}). 

    In particular, the trivial $\sigma$-algebra $\calF_0 = \{ \emptyset, \Omega \}$ denotes having no information at all and the expected value $\mathbb{E} \left[X\right]$ is the best approximation of $X: \Omega \rightarrow \mathbb{R}$ without any additional information.
    \end{remark}

    Similarly, we can compute the conditional expectations $Z_n = \mathbb{E} \left[ Z \mid \calF_n \right]$ for every $n \in \mathbb{N}$.

    \begin{lemma}
        \label{lemma: appendix example martingale}
        Let $Z: [0, 1) \rightarrow [0, 1]$ be defined as above and let $\{\calF_n\}_{n \in \mathbb{N}}$ be the filtration generated by dyadic intervals. 
        For $n \in \mathbb{N}$, let $Z_n := \mathbb{E} \left[ Z \mid \calF_n \right]$ be the conditional expectation. 
        Then for all $n \in \mathbb{N}$, $Z_n: [0, 1) \rightarrow [0, 1]$ is given by
        \begin{align*}
            Z_n(\omega) = 2^{-(n+1)} \sum_{k = 0}^{2^n - 1} (2k + 1) \cdot 1_{[k  2^{-n}, (k+1) 2^{-n})}(\omega).
        \end{align*}
        Note that $2^{-(n+1)} (2k + 1)$ is the midpoint of the interval $[k  2^{-n}, (k+1) 2^{-n})$.
     \end{lemma}

     \begin{proof}
         Fix $n \in \mathbb{N}$. 
         Then by definition $Z_n$ is $\calF_n$-measurable, i.e.~for all open $A \subseteq [0, 1]$, 
         \begin{align*}
             Z_n^{-1}(A) \in \calF_n 
             = \sigma \left(D_n \right) = \left\lbrace \bigcup_{k \in I} [k \cdot 2^{-n}, (k+1) \cdot 2^{-n}) : I \subseteq \{0, 1, \cdots, 2^n - 1 \} \right\rbrace.
         \end{align*}
         This implies that $Z_n$ is constant on every dyadic interval $[k 2^{-n}, (k+1) 2^{-n})$. 
         Therefore, there exist $c_{n, 0}, c_{n, 1}, \cdots, c_{n, 2^n - 1} \in \mathbb{R}$ such that for all $\omega \in [0, 1)$,
         \begin{align*}
             Z_n(\omega) = \sum_{k = 0}^{2^n - 1} c_{n, k} \cdot 1_{[k  2^{-n}, (k+1) 2^{-n})}(\omega). 
         \end{align*}
         Moreover, by \eqref{eq: second characteristic conditional expecatation} and by choosing $A = [k  2^{-n}, (k+1) 2^{-n}) \in \calF_n$ we obtain that 
         \begin{align*}
             c_{n, k} \cdot 2^{-n} 
             %= \int_{[k  2^{-n}, (k+1) 2^{-n})} c_{n, k} \ dx 
             &= \int_{[k  2^{-n}, (k+1) 2^{-n})} Z_n(x) \ dx       
             %= \int_{[k  2^{-n}, (k+1) 2^{-n})} Z(x) \ dx   
             = \int_{[k  2^{-n}, (k+1) 2^{-n})} x \ dx 
             = \frac{1}{2} \cdot 2^{-2n} \left((k+1)^2 - k^2 \right).
         \end{align*}
         This implies the claim.
     \end{proof}

     \begin{remark}
     \label{remark: second intuition conditional expectation}
     By Remark~\ref{remark: first intuition conditional expectation}, the trivial $\sigma$-algebra $\calF_0 = \{\emptyset, [0, 1)\}$ represents no information. 
     Likewise, if we view $Z$ as a random number generator on $[0, 1]$, then for every $n \in \mathbb{N}$, the $\sigma$-algebra $\calF_n = \sigma(D_n)$ restricts the information as to only being able to know in which dyadic interval 
     $[0, 2^{-n}), [2^{-n}, 2^{1-n}), \cdots, [1 - 2^{-n}, 1)$ the sampled value $Z(\omega)$ lies. 
     \end{remark}

     Furthermore, there are some useful computational rules when working with conditional expectations. 
     We recommend for instance \cite{klenke,Williams91} for further reading.

     \begin{theorem}
     \label{thm: properties conditional expectation}
         Let $(\Omega, \calF, \mathbb{P})$ be a probability space and let $X, Y: \Omega \rightarrow \mathbb{R}$ be two integrable random variables. 
         Moreover, let $\calG \subseteq \calH \subseteq \calF$ be $\sigma$-algebras and let $a, b, c \in \mathbb{R}$. 
         Then the following holds. 
         \begin{enumerate}
             \item\label{item: ce 1} 
             Preservation of expectation: $\mathbb{E} \left[ \mathbb{E} \left[X \mid \calG \right] \right] = \mathbb{E} \left[X\right]$. 
             
             \item\label{item: ce 1,5} 
             Trivial $\sigma$-algebra: $\mathbb{E} \left[ \mathbb{E} \left[X \mid \{\emptyset, \Omega \} \right] \right] = \mathbb{E} \left[X\right]$.

             \item\label{item: ce 2} 
             Preservation of constants: $\mathbb{E} \left[c \mid \calG \right] = c$ almost surely.

             \item\label{item: ce 3} 
             Linearity: $\mathbb{E} \left[a X + b Y \mid \calG \right] = a \cdot \mathbb{E} \left[X \mid \calG \right] + b \cdot \mathbb{E} \left[Y \mid \calG \right]$ almost surely.

             \item\label{item: ce 4} 
             Taking out what is known: 
             If the product $XY$ is integrable and $X$ is $\calG$-measurable, then we have $\mathbb{E} \left[X Y \mid \calG \right] = X \cdot \mathbb{E} \left[Y \mid \calG \right]$ almost surely. 

             \item\label{item: ce 5} 
             Tower property: $\mathbb{E} \left[\mathbb{E} \left[X \mid \calG \right] \mid \calH \right] = \mathbb{E} \left[X \mid \calG \right] = \mathbb{E} \left[\mathbb{E} \left[X \mid \calH \right] \mid \calG \right]$.

            \item\label{item: ce 6} 
            Independence: If $\sigma(X)$ and $\calG$ are independent\footnote{
                This means that $\mathbb{P} \left( X^{-1}(A) \cap B \right) = \mathbb{P} \left(X^{-1}(A)\right) \cdot \mathbb{P} \left(B\right)$ for all (Borel) measurable $A \subseteq \mathbb{R}$ and all $B \in \calG$. 
                This holds, if $\calG = \sigma(Y)$ is the smallest $\sigma$-algebra such that $Y$ is $\calG$-measurable and $X, Y$ are independent random variables. This is generally the case in this paper, as we assume that the random samples are i.i.d.~over time.
             }, 
             then $\mathbb{E} \left[X \mid \calG \right] = \mathbb{E} \left[X\right]$. 

            \item\label{item: ce 7} 
            Monotonicity: If $X \leq Y$ almost surely, then $\mathbb{E} \left[X \mid \calG \right] \leq \mathbb{E} \left[Y \mid \calG \right]$ almost surely.

            \item\label{item: ce 8} 
             Triangle inequality: $\left| \mathbb{E} \left[X \mid \calG \right] \right| \leq \mathbb{E} \left[|X| \mid \calG \right]$.
         \end{enumerate}
     \end{theorem}

     \begin{proof}
        For property~\eqref{item: ce 1} note that because $\Omega \in \calG$, as $\calG$ is a $\sigma$-algebra, and by the definition of the (conditional) expectation:
        \begin{align*}
            \mathbb{E} \left[ \mathbb{E} \left[X \mid \calG \right] \right] 
            = \mathbb{E} \left[ 1_{\Omega} \cdot \mathbb{E} \left[X \mid \calG \right] \right] 
            %= \int_{\Omega} \mathbb{E} \left[X \mid \calG \right] \ d\mathbb{P}
            %= \int_{\Omega} X \ d\mathbb{P} 
            = \mathbb{E} \left[ 1_{\Omega} \cdot X \right] 
            = \mathbb{E} \left[X \right]. 
        \end{align*}
        This proves~\eqref{item: ce 1}. 
        We have proven \eqref{item: ce 1,5} in Remark~\ref{remark: first intuition conditional expectation}. 
         It is also true by \cite[Theorem 8.14]{klenke}.

        For \eqref{item: ce 2} note that constant functions are $\{ \emptyset, \Omega \}$-measurable. 
        Hence as $\{ \emptyset, \Omega \} \subseteq \calG$, constant functions are also $\calG$-measurable. 
        Therefore, in order to prove~\eqref{item: ce 2} it is sufficient to show that 
        \begin{align*}
            \mathbb{E} \left[ 1_A \cdot \mathbb{E} \left[ c \mid \calG \right] \right]
            = 
            \mathbb{E} \left[ 1_A \cdot c \right] = c \cdot \mathbb{P}(A) 
            %\int_A c \ d\mathbb{P} \overset{!}{=} \int_A \mathbb{E} \left[c \mid \calG \right] \ d\mathbb{P} = \int_A c \ d\mathbb{P} = c \cdot \mathbb{P}(A) 
            \qquad \mathrm{for \ all \ } A \in \calG.
        \end{align*}
        This is trivial. Thus, property~\eqref{item: ce 2} holds.
     
        Lastly, the properties \eqref{item: ce 3}, \eqref{item: ce 4}, \eqref{item: ce 5}, \eqref{item: ce 6}, \eqref{item: ce 7}~\&~\eqref{item: ce 8} hold by \cite[Theorem 8.14]{klenke}. 
     \end{proof}

     \paragraph{Martingales} 
     Recall that a sequence of integrable random variables $\{M_n\}_{n \in \mathbb{N}}$ is \emph{adapted} to a filtration $\{\calF_n\}_{n \in \mathbb{N}}$, if $M_n$ is $\calF_n$-measurable for all $n \in \mathbb{N}$.
     Given a filtration $\{\calF_n\}_{n \in \mathbb{N}}$, an integrable, $\{\calF_n\}_{n \in \mathbb{N}}$-adapted sequence of random variables $\{M_n\}_{n \in \mathbb{N}}$ is 
	\begin{enumerate}[(i)]
		\item a {\em martingale} if for all $n \in \mathbb{N}$,
		$\mathbb{E}[M_{n+1} \mid \calF_n] = M_n$ almost surely (a.s.);
		\item a	{\em supermartingale} if for all $n \in \mathbb{N}$,
		$\mathbb{E}[M_{n+1} \mid \calF_n] \leq M_n$ a.s.;
		\item a	{\em submartingale}	if for all $n \in \mathbb{N}$,
		$\mathbb{E}[M_{n+1} \mid \calF_n] \geq M_n$ a.s..
	\end{enumerate}
    In particular, if $\{M_n\}_{n \in \mathbb{N}}$ is a martingale, then $\mathbb{E} \left[M_n \right] = \mathbb{E} \left[M_0 \right]$ for every $n \in \mathbb{N}$ by Theorem~\ref{thm: properties conditional expectation}. 
    Similarly, if $\{M_n\}_{n \in \mathbb{N}}$ is a supermartingale, resp.~submartingale, then we have $\mathbb{E} \left[M_n \right] \leq \mathbb{E} \left[M_0 \right]$, resp.~$\mathbb{E} \left[M_n \right] \geq \mathbb{E} \left[M_0 \right]$, for all $n \in \mathbb{N}$. 

    \begin{remark}
        By the Remarks~\ref{remark: first intuition conditional expectation}~\&~\ref{remark: second intuition conditional expectation}, 
        intuitively a martingale $\{M_n\}_{n \in \mathbb{N}}$ is information preserving under the filtration $\{\calF_n\}_{n \in \mathbb{N}}$. 
        This makes them natural candidates for probabilistic invariants.
        Specifically, the random variable $M_n: \Omega \rightarrow \mathbb{R}$ is the best approximation of $M_{n+1}: \Omega \rightarrow \mathbb{R}$ when restricted to the information represented by $\calF_n \subseteq \calF_{n+1}$. 
        Similarly, for a supermartingale, resp.~a submartingale, the random variable $M_{n}: \Omega \rightarrow \mathbb{R}$ is allowed to overapproximate, resp.~underapproximate, $M_{n+1}: \Omega \rightarrow \mathbb{R}$ given the information represented by $\calF_n$. 
    \end{remark}

    \begin{lemma}
        The family of random variables $\{Z_n\}_{n \in \mathbb{N}}$ from Lemma~\ref{lemma: appendix example martingale} is a martingale 
        with respect to the filtration $\{\calF_n\}_{n \in \mathbb{N}}$ generated by dyadic intervals.
    \end{lemma}

    \begin{proof}
        By Lemma~\ref{lemma: appendix example martingale} it is sufficient to show that $\{Z_n\}_{n \in \mathbb{N}}$ satisfies the martingale property, i.e.
        \begin{align*}
            \mathbb{E}[Z_{n+1} \mid \calF_n] = Z_n
        \end{align*}
        almost surely for all $n \in \mathbb{N}$. 

        Hence, let $n \in \mathbb{N}$ be fixed. As in the proof of Lemma~\ref{lemma: appendix example martingale} the random variable 
        $\mathbb{E}[Z_{n+1} \mid \calF_n]$ is constant on all dyadic intervals $[k 2^{-n}, (k+1) 2^{-n})$ for $k \in \{0, 1, \cdots, 2^n - 1\}$. 
        Therefore, there exist numbers $c_{n, 0}, c_{n, 1}, \cdots, c_{n, 2^n - 1} \in \mathbb{R}$ such that for all $\omega \in [0, 1)$,
         \begin{align*}
             \mathbb{E}[Z_{n+1} \mid \calF_n](\omega) = \sum_{k = 0}^{2^n - 1} c_{n, k} \cdot 1_{[k  2^{-n}, (k+1) 2^{-n})}(\omega). 
         \end{align*}
         So, by choosing $A = [k  2^{-n}, (k+1) 2^{-n}) \in \calF_n$ and by Lemma~\ref{lemma: appendix example martingale} we obtain that 
         \begin{align*}
             c_{n, k} \cdot 2^{-n} 
             %= \int_{[k  2^{-n}, (k+1) 2^{-n})} c_{n, k} \ dx 
             &= \int_{[k  2^{-n}, (k+1) 2^{-n})} Z_{n+1}(x) \ dx    
             \\ 
             &= 2^{-(n+2)} \int_{k  2^{-n}}^{(2k+1) 2^{-(n+1)}} (4k + 1) \ dx 
             + 2^{-(n+2)} \int_{(2k+1) 2^{-(n+1)}}^{(k+1) 2^{-n}} (4k + 3) \ dx 
             \\
             &= 2^{-(2n+3)} \cdot \left( 8k + 4 \right) = 2^{-(2n+1)} \cdot \left( 2k + 1 \right).
         \end{align*}
         This implies that $c_{n, k} = 2^{-(n+1)} \cdot \left( 2k + 1 \right)$. 
         This concludes the proof by Lemma~\ref{lemma: appendix example martingale}.
    \end{proof}

    \clearpage

    \section{Known versions of the Optional Stopping Theorem}
    \label{sec: appendix ost}

    In this section we present different versions of the Optional Stopping Theorem that are known in the mathematical literature or have been proven in more recent Computer Science papers. 

    ~

    First, the Optional Stopping Theorem holds, if either the stopping time or the martingale is bounded.
    
    \begin{theorem}[{e.g.~\cite[Theorem 10.10]{Williams91}}]
    \label{thm: ost for bounded martingales}
        Let $(\Omega, \calF, \mathbb{P})$ be a probability space, and $M = \{M_n\}_{n \in \mathbb{N}}$ be a supermartingale adapted to a filtration $\{\calF_n\}_{n \in \mathbb{N}}$. 
        Let $T: \Omega \rightarrow \mathbb{N} \cup \{\infty\}$\footnote{
            \label{footnote: diverging runs appendix}
            Note that as in Remark~\ref{remark: infinite stopping time}, if $\mathbb{P}(T = \infty) > 0$, 
            then on the event $\{T = \infty\}$ there exists an integrable limit function $M_{\infty}$ such that $\lim_{k \rightarrow \infty} M_k = M_{\infty}$ almost surely. 
        } be a stopping time. 
        Then $M_T$ is integrable and 
        \begin{align*}
            \mathbb{E} \left( M_T \right) \leq \mathbb{E} \left(M_0 \right)
        \end{align*}
        in each of the following situations:
        \begin{enumerate}
            \item $T$ is bounded, i.e.~$\mathbb{P} \left( T \leq N \right) = 1$ for some $N \in \mathbb{N}$;
            \item $M$ is bounded, i.e.~$\mathbb{P} \left( \left| M_n \right| \leq K \right) = 1$ holds for some $K > 0$ and all $n \in \mathbb{N}$;
            %\item $\mathbb{E} \left( T \right) < \infty$
        \end{enumerate}
    \end{theorem}

    Second, the Optional Stopping Theorem also holds, if both the stopping time and the martingale are bounded in some sense. 
    
    \begin{theorem}[{e.g.~\cite[Theorem 12.5.9]{grimmett}}]
    \label{thm: first mixed ost}
        Let $(\Omega, \calF, \mathbb{P})$ be a probability space, and $M = \{M_n\}_{n \in \mathbb{N}}$ be a martingale adapted to a filtration $\{\calF_n\}_{n \in \mathbb{N}}$. 
        Let $T: \Omega \rightarrow \mathbb{N} \cup \{\infty\}$ be a stopping time. 
        Then $\mathbb{E} \left( M_T \right) = \mathbb{E} \left(M_0 \right)$, if the following conditions hold.
        \begin{enumerate}
            \item $\mathbb{E} \left( T \right) < \infty$, 
            \item there is a constant $c \in (0, \infty)$ 
            such that $\mathbb{E} \left( \left| M_{n + 1} - M_{n} \right| \mid \calF_n \right) \leq c$ holds almost surely on $\{T \geq n\}$ for all $n \in \mathbb{N}$. 
        \end{enumerate}
    \end{theorem}

    \begin{theorem}[{e.g.~\cite[Theorem 12.5.1]{grimmett}}]
    \label{thm: weird mixed ost}
        Let $(\Omega, \calF, \mathbb{P})$ be a probability space, and $M = \{M_n\}_{n \in \mathbb{N}}$ be a martingale adapted to a filtration $\{\calF_n\}_{n \in \mathbb{N}}$. Let $T: \Omega \rightarrow \mathbb{N} \cup \{\infty\}$ be a stopping time. 
        Then $\mathbb{E} \left( M_T \right) = \mathbb{E} \left(M_0 \right)$, if the following conditions hold.
        \begin{enumerate}
            \item $\mathbb{P} \left( T < \infty \right)  = 1$, 
            \item $\mathbb{E} \left( \left| M_T \right| \right) < \infty$,
            \item and 
                $\displaystyle \lim_{n \rightarrow \infty} \mathbb{E} \left( M_n \cdot 1_{\{ T > n \}} \right) = 0$.
        \end{enumerate}
    \end{theorem}

    An important special case is that non-negative supermartingales satisfy the Optional Stopping Theorem automatically. 
    
    \begin{theorem}[{e.g.~\cite[Theorem 4.6.7]{cohen}}]
    \label{thm: ost for non-negative supermartingales}
        Let $(\Omega, \calF, \mathbb{P})$ be a probability space, let $M = \{M_n\}_{n \in \mathbb{N}}$ be a non-negative supermartingale adapted to a filtration $\{\calF_n\}_{n \in \mathbb{N}}$, and let $T: \Omega \rightarrow \mathbb{N} \cup \{\infty\}$\footref{footnote: diverging runs appendix} be a stopping time. 
        Then $\mathbb{E} \left( M_T \right) \leq \mathbb{E} \left(M_0 \right)$.
    \end{theorem}

    Moreover, a uniformly integrable martingale satisfies the Optional Stopping Theorem for every stopping time. 

    \begin{definition}
        Let $(\Omega, \calF, \mathbb{P})$ be a probability space and let $M = \{M_n\}_{n \in \mathbb{N}}$ be a martingale adapted to a filtration $\{\calF_n\}_{n \in \mathbb{N}}$.  
        Then we say that $M = \{M_n\}_{n \in \mathbb{N}}$ is {\em uniformly integrable}, if for all $\varepsilon > 0$ there exists a $K > 0$ such that 
        \begin{align*}
            \mathbb{E} \left( \left| M_n \right| \cdot 1_{\{ |M_n| > K \}} \right) < \varepsilon
        \end{align*}
        holds for all $n \in \mathbb{N}$. 
    \end{definition}
    
    \begin{theorem}[{e.g.~\cite[Theorem 4.6.7]{cohen}}]
    \label{thm: ost for ui martingales}
        Let $(\Omega, \calF, \mathbb{P})$ be a probability space, and let the process $M = \{M_n\}_{n \in \mathbb{N}}$ be a uniformly integrable martingale adapted to a filtration $\{\calF_n\}_{n \in \mathbb{N}}$. 
        Moreover, let $T: \Omega \rightarrow \mathbb{N} \cup \{\infty\}$\footnote{
            \label{footnote: diverging runs appendix 2}
            Note that as in Remark~\ref{remark: infinite stopping time}, if $\mathbb{P}(T = \infty) > 0$, 
            then on the event $\{T = \infty\}$ there exists an integrable limit function $M_{\infty}$ such that $\lim_{k \rightarrow \infty} M_k = M_{\infty}$ almost surely. 
        } be a stopping time. 
        Then $\mathbb{E} \left( M_T \right) = \mathbb{E} \left(M_0 \right)$.
    \end{theorem}

    Based on this, two new versions of the Optional Stopping Theorem have been proven. 
    The first one was shown in \cite{Wang2021ExpectedCostAF} and has successfully been applied in \cite{wang2021}.

    \begin{theorem}[{\cite[Theorem 3.12]{Wang2021ExpectedCostAF}}]
        Let $(\Omega, \calF, \mathbb{P})$ be a probability space, and let $M = \{M_n\}_{n \in \mathbb{N}}$ be a (sub/super-)martingale adapted to a filtration $\{\calF_n\}_{n \in \mathbb{N}}$.         
        Let $T: \Omega \rightarrow \mathbb{N} \cup \{\infty\}$ be a stopping time. 
        Then $\{M^T_n\}_{n \in \mathbb{N}}$, i.e.~$M_n^T := M_{\min(n, T)}$, is uniformly integrable, if $\mathbb{E} \left( p(T) \right) < \infty$ and $|M_n| \leq p(n)$ a.s.~on $\{T \geq n\}$, holds for all $n \in \mathbb{N}$, for some non-decreasing function $p$.
    \end{theorem}

    Specifically, the following corollary is applied in \cite{wang2021}.

    \begin{theorem}[{\cite[Corollary 3.13]{Wang2021ExpectedCostAF}}]
        Let $(\Omega, \calF, \mathbb{P})$ be a probability space, and let $M = \{M_n\}_{n \in \mathbb{N}}$ be a (sub/super-)martingale adapted to a filtration $\{\calF_n\}_{n \in \mathbb{N}}$.         
        Let $T: \Omega \rightarrow \mathbb{N} \cup \{\infty\}$ be a stopping time. 
        Then $\{M^T_n\}_{n \in \mathbb{N}}$, i.e.~$M_n^T := M_{\min(n, T)}$, is uniformly integrable, if there exist $l \in \mathbb{N}$ and $C \geq 0$ such that $\mathbb{E} \left( T^l \right) < \infty$ and for all $n \in \mathbb{N}$, $|M_n| \leq C \cdot (n + 1)^l$ a.s.~on $\{T \geq n\}$. 
    \end{theorem}
    
    \begin{theorem}[{\cite[Theorem 5.2]{costanalysis}}] 
        Let $(\Omega, \calF, \mathbb{P})$ be a probability space, and let $M = \{M_n\}_{n \in \mathbb{N}}$ be a martingale (resp.~supermartingale) adapted to a filtration $\{\calF_n\}_{n \in \mathbb{N}}$.         
        Let $T$ be a stopping time. 
        Then the following conditions are sufficient to ensure that $\mathbb{E} \left( |M_T| \right) < \infty$ 
        and $\mathbb{E} \left( M_T \right) = \mathbb{E} \left( M_0 \right)$ (resp.~$\mathbb{E} \left( M_T \right) \leq \mathbb{E} \left( M_0 \right)$): 
        \begin{enumerate}
            \item There exist constants $c_1, c_2 > 0$ such that  
            $\mathbb{P} \left( T > n \right) \leq c_1 \cdot e^{-c_2 \cdot n}$ holds for sufficiently large $n \in \mathbb{N}$. 
            \item 
            There exist constants $K, d > 0$ such that for all $n \in \mathbb{N}$, $\left| M_{n + 1} - M_n \right| \leq K \cdot n^d$ holds almost surely. 
        \end{enumerate}
    \end{theorem}

    The latter result has been strengthened in \cite{wang2024static}. 
    This stronger version is applied in \cite{piecewise_analysis_prob2026}.

    \begin{theorem}[{\cite[Theorem 4.5]{wang2024static}}]
        Let $(\Omega, \calF, \mathbb{P})$ be a probability space, and let $M = \{M_n\}_{n \in \mathbb{N}}$ be a supermartingale adapted to a filtration $\{\calF_n\}_{n \in \mathbb{N}}$.        
        Let $T$ be a stopping time w.r.t.~$\{\calF_n\}_{n = 0}^{\infty}$. 
        Suppose that there exist positive real numbers $b_1, b_2, c_1, c_2, c_3$ such that $c_2 > c_3$ and 
        \begin{description}
            \item[A1] $\mathbb{P} \left( T > n \right) \leq c_1 \cdot e^{-c_2 \cdot n}$ for sufficiently large $n \in \mathbb{N}$, and 
            \item[A2] for all $n \in \mathbb{N}$, $\left| M_{n+1} - M_n \right| \leq b_1 \cdot n^{b_2} \cdot e^{c_3 \cdot n}$ holds almost surely.
        \end{description}
        Then we have that $\mathbb{E} \left( \left| M_T \right| \right) < \infty$ and $\mathbb{E} \left( M_T \right) \leq \mathbb{E} \left( M_0 \right)$.
    \end{theorem}

    \clearpage

    \section{Proof of the \dui-precondition}
	\label{sec: background background}

\thmOptionalStopping*
	\begin{proof}
		The stopped-at-$T$ process, $\{M_{n \wedge T}\}_{n \in \mathbb{N}}$, is a (super)martingale. 
		%if $\{M_{n}\}_{n \in \mathbb{N}}$ is a (super)martingale.
		Since $X$ is integrable, the family $\{X\} \subseteq L^1$ is uniformly integrable by \cite[Theorem~6.18(i)]{klenke}.
		It follows (from \cite[Theorem 6.18(iii)]{klenke}) that $\{M_{n \wedge T}\}_{n \in \mathbb{N}}$ is a uniformly integrable (super)martingale.
		By \cite[Theorem 11.7]{klenke} there is an integrable random variable $M_\infty$ with $\lim_{n \to \infty} M_{n \wedge T} = M_\infty$ a.s.~and in $L^1$. 
		Since $M_{n \wedge T}$ converges %point-wise 
		to $M_{T}$ in distribution, we have $M_T = M_\infty$. 
		It then follows from convergence in $L^1$ that
		\(
		\lim_{n \rightarrow \infty} \mathbb{E}(M_{n \wedge T}) = \mathbb{E}(M_{T}).
		\)	
		Thus, by \cite[p.~99]{Williams91}, 
		in case $\{M_{n \wedge T}\}_{n \in \mathbb{N}}$ is a martingale: for all $n \in \mathbb{N}$, 
		$\mathbb{E}(M_{0}) = \mathbb{E}(M_{n \wedge T})$,
		and so $\mathbb{E}(M_{0}) = \mathbb{E}(M_{T})$; in case $\{M_{n \wedge T}\}_{n \in \mathbb{N}}$ is a supermartingale: for all $n \in \mathbb{N}$, 
		$\mathbb{E}(M_{n \wedge T})\leq \mathbb{E}(M_{0})$.
		Limits preserve inequalities, so
		\(
        \mathbb{E}(M_{T})=\lim_{n \rightarrow \infty} \mathbb{E}(M_{n \wedge T}) \leq \lim_{n \rightarrow \infty} \mathbb{E}(M_{0}) =\mathbb{E}(M_{0}).
		\)	
	\end{proof}

	\bigskip{}

	\section{Operational Semantics and the proofs of Section~\ref{sec: program}}
    \label{sec: operational semantics}
    
	\begin{definition}[Operational Semantics]\rm 
		\label{def: operational semantics PTS}%~\\
		%\changed[lo]{
		The {\em operational semantics} of a PTS $\Pi$ is a sequence of random variables $\{L^{[k]}, X^{[k]}\}_{k \in \mathbb{N}}$, where $L^{[k]}: \Omega \rightarrow L$ and $X^{[k]} : \Omega \rightarrow \mathbb{R}^n$ satisfy
		\begin{enumerate}[(i)]
			\item  
			$L^{[0]}(\omega) = l_0$ is the initial location for all $\omega \in \Omega$, 
			\item
			$X^{[0]}(\omega) := \vec{x}_0 \in \mathbb{R}^n$ for all $\omega \in \Omega$, where $\vec{x}_0$ is a vector containing the initial values of the program variables $X = \{x_1, x_2, \cdots, x_n\}$,  
			%is a random variable with distribution $D_0$ as defined in Def.~\ref{def: pts},
			\item 
			Let $\omega \in \Omega$, $k \in \mathbb{N}$ and let $\tau(\omega)$ be the unique transition which is enabled at the configuration $(L^{[k]}(\omega), X^{[k]}(\omega))$.
            Then $L^{[k+1]}(\omega) = \dloc{\tau(\omega)}$ and 
			$X^{[k+1]}(\omega) = f_{\tau(\omega)} (X^{[k]}(\omega), R^{[k+1]}(\omega))$.
		\end{enumerate}
		where $R^{[k]}(\omega)$ are the 
%		vector of 
		random samples which are sampled in the $k$-th step of the computation.
    \end{definition}

	We note that  $\{L^{[k]}, X^{[k]}\}_{k \in \mathbb{N}}$ is adapted to the (natural) filtration $\{\calF_k\}_{k \in \mathbb{N}}$, i.e.~$\calF_k$ is the smallest $\sigma$-algebra such that the random variables $X^{[0]}, R^{[1]}, \cdots, R^{[k]}$ are adapted to $\calF_k$. 
    This is denoted as $\calF_k := \sigma(X^{[0]}, R^{[1]}, \cdots, R^{[k]})$.  
     
    Intuitively, $\calF_k$ contains information on the initial value $\left(l_0, \vec{x}_0 \right)$ and the first $k$ random samples. %  $R^1, R^2, \cdots, R^n$. 
	In other words, $\calF_{k}$ stores the information of the program run up to the $k$-th step of the computation.
	In particular,  
	$\left\lbrace L^{[k]}, X^{[k]} \right\rbrace_{k \in \mathbb{N}}$ being adapted to ${\calF_k}$ states that $\left\lbrace L^{[k]}, X^{[k]} \right\rbrace_{k \in \mathbb{N}}$ is fully determined by the information held by $\calF_{k}$.

	\begin{lemma}
		\label{lemma: pts,measurable}
		%~\\
		Let $\Pi$ be a PTS and let $\{L^{[k]}, X^{[k]}\}_{k \in \mathbb{N}}$ be the operational semantics of $\Pi$ and let $\calF_k = \sigma(X^{[0]}, R^{[1]}, \cdots, R^{[k]})$ be our standard filtration. 
		Then the operational semantics 
        %$\{X^{[k]}, L^{[k]}\}_{k \in \mathbb{N}}$ 
        is $\mathcal{P}(L) \otimes \mathbb{B}(\mathbb{R}^n)$-$\calF_k$-measurable\footnote{
			Here $\mathcal{P}$ denotes the powerset and $\mathbb{B}(\mathbb{R}^n)$ denotes the Borel $\sigma$-algebra on $\mathbb{R}^n$.
		}.
	\end{lemma}
	
	\begin{proof} %~\\
		%As $\Pi$ is non-demonic, there is exactly one transition enabled for each $(x_n, l_n) \in D \times L$.
		%Hence $\tau_{k+1}$ is a $\mathbb{B}(D) \otimes P(L)$-measurable function.
		We prove the claim by induction.
		
		For the base case $k$ = 0:
		Because $\calF_0 = \sigma(X^{[0]})$,  $X^{[0]}$ is $\mathbb{B}(\mathbb{R}^n)$-$\calF_0$-measurable by definition. 
		As a constant function, $L^{[0]}$ is trivially $\mathcal{P}(L)$-$\calF_0$-measurable.
		
		For the inductive case, we have, by the induction hypothesis that 
		\begin{center}
			\begin{tabular}{lll}
				$g_k$: &$\Omega$ &$\rightarrow L \times \mathbb{R}^n$\\
				&$\omega$ &$\mapsto (L^{[k]}(\omega), X^{[k]}(\omega))$
			\end{tabular}
		\end{center}
		is $\mathbb{B}(\mathbb{R}^n) \otimes \mathcal{P}(L)$-$\calF_k$-measurable.
	
		In particular, $\tau_{k+1}$ is uniquely determined by $L^{[k]}, X^{[k]}$ based on polynomial guards on $\mathbb{R}^n$.
		This can be represented by a $\mathcal{P}(L) \otimes \mathbb{B}(\mathbb{R}^n)$-$\mathcal{P}(L)$-mea\-sur\-able selection function $\psi$.
		Hence, the composition $\tau_{k+1} = \psi \circ g_k$ is 
		$\mathcal{P}(L)$-$\calF_k$-measurable. 
        Moreover, 
		$L^{[k+1]}$ is uniquely defined by the transition $\tau_{k+1}$ and $L^{[k]}$. 
        Hence, the random variable $L^{[k+1]}$ is $\mathcal{P}(L)$-$\calF_k$-measurable  
		%So $L^{[k+1]}$ is $\mathcal{P}(L)$-$\calF_{k+1}$-measurable 
        and therefore, as $\calF_k \subseteq \calF_{k + 1}$, it is also $\mathcal{P}(L)$-$\calF_{k+1}$-measurable.
        Furthermore, 
		$X^{[k+1]} = f_{\tau_{k+1}}(X^{[k]}, R^{[k+1]})$ by Definition~\ref{def: pts} and 
		the function $f_{\tau_{k+1}}$ is continuous.
		Thus, this implies that the random variable $X^{[k+1]}$ is
		$\mathbb{B}(\mathbb{R}^{n})$-$\calF_{k+1}$-measurable.
		%Note that $L^{[k+1]} = l_{\tau_{k+1}}$ does only depend on $\tau_{k+1}$.
        
		In conclusion, $(L^{[k+1]}, X^{[k+1]})$ is $\mathcal{P}(L) \otimes \mathbb{B}(\mathbb{R}^{n})$-$\calF_{k+1}$-measurable as claimed. 
		%\qed		
	\end{proof}
	
	\lemmatilde*
	
	\begin{proof} 
		Let $k \in \mathbb{N}$. Note that $\calF_k = \sigma(X^{[0]}, R^{[1]}, R^{[2]}, \cdots, R^{[k]})$.  
        Therefore, by Theorem~\ref{thm: properties conditional expectation},
		\[\begin{array}{cl}
		%\begin{split}
		&\mathbb{E}[h(L^{[k+1]}, X^{[k+1]}, k+1) \mid \calF_k]
		\\
		= & 
		\mathbb{E}\Big[\Big( \sum_{i = 1}^p [(L^{[k]} = \sloc{\tau_i}) \wedge (X^{[k]} \models \phi_i)] \Big) \cdot h(L^{[k+1]}, X^{[k+1]}, k+1) \mid \calF_k \Big]
		\\
		= &
		\mathbb{E}\Big[\Big( \sum_{i = 1}^p [(L^{[k]} = \sloc{\tau_i}) \wedge (X^{[k]} \models \phi_i)] \Big) \cdot h(\dloc{\tau_{i}}, f_{\tau_i}(X^{[k]}, R^{[k+1]}), k+1) \mid \calF_k \Big] 
		\\
		= &
		\sum_{i = 1}^p [(L^{[k]} = \sloc{\tau_i}) \wedge (X^{[k]} \models \phi_i)]  \cdot
		\mathbb{E}\Big[  h(\dloc{\tau_{i}}, f_{\tau_i}(X^{[k]}, R^{[k+1]}), k+1) \mid \calF_k \Big] 
		\\
		&\hspace*{12pt} \mathrm{ \ as \ }
		X^{[k]}, L^{[k]} \mathrm{\ are \ } \calF_{k}\mathrm{-measurable}
		%\end{split}
		\end{array}\]
		Moreover, by assumption $h_{\dloc{\tau_i}}: \mathbb{R}^n \times \mathbb{N} \rightarrow \mathbb{R}$ defined by $h_{\dloc{\tau_i}} \left( \vec{x}, k \right) := h(\dloc{\tau_i}, \vec{x}, k)$ is a polynomial. 
        The random samples are independent of $\calF_k$ and the random variable $X^{[k]}$ is $\calF_k$-measurable. 
        Consequently, by \cite[9.10 {(a)} and {(b)}, page 91 and 92]{Williams91} we can compute that 
		\[\begin{array}{ll}
		%\begin{split}
		 \mathbb{E}\Big[  h(\dloc{\tau_{i}}, f_{\tau_i}(X^{[k]}, R^{[k+1]}), k+1) \mid \calF_k \Big]
%		\\
		&= 
		\mathbb{E}\Big[ 
		h_{\dloc{\tau_i}}(f_{\tau_i}(X^{[k]}, R^{[k+1]}),  k+1) \mid \calF_k \Big]
		\\
		&= 
		\mathbb{E}_R\Big[  h_{\dloc{\tau_i}}(f_{\tau_i}(X^{[k]}, R),  k+1) \Big],
		\end{array}
		\]
		where $X^{[k]}$ is treated like a constant (by virtue of being $\calF_k$-measurable) and the expectation is taken over the distributions of $R^{[k+1]} \sim R^{[1]}$. 
		This proves \eqref{item: lemma tilde i} by the definition of the pre-expectation and by Theorem~\ref{thm: properties conditional expectation}.	
		
		As for part \eqref{item: lemma tilde ii}, $\mathbb{E}[h(L^{[k+1]}, X^{[k+1]}, k+1) \mid \calF_k] = \mathrm{pre}\mathbb{E}(h)(L^{[k]}, X^{[k]}, k+1) \leq h(L^{[k]}, X^{[k]}, k)$ holds by part \eqref{item: lemma tilde i}. 
        This concludes the proof. 
		%\qed
	\end{proof}

    \bigskip{}
    
 	\section{Proofs of Section~\ref{sec: linear}}
    \label{sec: proofs linear}
 	
	The norm $\| \ \|_2$ on $\mathbb{R}^n$ is defined by $\|x\|_2 = \sqrt{\sum_{i = 1}^n x_i^2}$.
	This vector norm induces a norm on  matrices $A \in \mathbb{R}^{n \times n}$ 
	by $\|A\|_2 := \sup_{x \neq 0} \frac{\|A x\|_2}{\|x\|_2}$.
	In particular, this matrix norm is sub-multiplicative, i.e.~$\| A B \|_2\leq \| A \|_2 \cdot \|B\|_2$ for two matrices $A, B \in \mathbb{R}^{n \times n}$.

    ~
    
	For this section we fix a linear PTS as defined in Definition~\ref{def:triangular-linear-update-loop}. 

    %\thmdominatingfunctionlinearpts*
    
    \begin{theorem}[Almost Sure Bounds]
        \label{thm:dominating_function_for_alpha_1:appendix}
        %Define the random variable $Z_i$, and 
        Fix a program variable $x_i$. Let  $A_i, B_i$ be the update matrices for the variables $x_i$ depends on. Let $A_i=P_iJ_iP_i^{-1}$ be a Jordan normal form of $A_i$, and $m_i$ be the size of the largest Jordan block of $J_i$:

        \begin{align*}
            Z_i:= \delta_i(T^{m_i-1}+\gamma_i)\|\vec{x}^{[0]}\|_2+\delta_i(T^{m_i-1}+\gamma_i)\beta_i\sum_{l=1}^T \|\vec{r}^{[l]}\|_2
        \end{align*}
        with  $\delta_i := a^{(m_i-1)}_i m_i \|P_i\|_2\|P_i^{-1}\|_2$, $\beta_i=\|B_i\|_2$, $\gamma_i:=(2(m_i-1))^{m_i-1}$ and $a_i$ being the largest absolute value within $J_i$ (but at least $1$).
        Then we have for all $k \in \mathbb{N}$ almost surely,
        \begin{align*}
            0 \leq \left| x_i^{[\min(k, T)]} \right| \leq Z_i,
        \end{align*}
        given that the maximum absolute value of an eigenvalue of $A$ is at most $1$. Notably, $\beta_i=0$ when $x_i$ does not depend on any random variable.
    \end{theorem}

    \begin{proof}
        By definition the program variables satisfy
        \[x^{[k]}=A^kx^{[0]}+\sum_{l=1}^{k} A^{k-l}(Br^{[l]}).\]
        We denote now with $y^{[k]}$ the sub-vector of $x^{[k]}$, which consists of the variables on which $x_i$ depends, 
        and obtain
        \[y^{[k]}=A_i^ky^{[0]}+\sum_{l=1}^{k} A_i^{k-l}(B_ir^{[l]}).\]
        Using the Jordan normal form $A_i=P_iJ_iP_i^{-1}$, this becomes
        \[y^{[k]}=P_iJ_i^kP_i^{-1}y^{[0]}+\sum_{l=1}^k P_iJ_i^{k-l}P_i^{-1}(B_ir^{[l]}).\]
        Notice that since $x_i^{[k]}$ is an element of $y^{[k]}$, and $y^{[0]}$ is a subvector of $x^{[0]}$, it holds that
        \begin{align*}
        |x_i^{[k]}|=\|x_i^{[k]}\|_2&\leq \|y^{[k]}\|_2=\|P_iJ_i^kP_i^{-1}y^{[0]}+\sum_{l=1}^k P_iJ_i^{k-l}P_i^{-1}(B_ir^{[l]})\|_2\\
        &\leq\|P_i\|_2\|J_i^k\|_2\|P_i^{-1}\|_2\|y^{[0]}\|_2+\sum_{l=1}^k \|P_i\|_2\|J_i^{k-l}\|_2\|P_i^{-1}\|_2\|B_ir^{[l]}\|_2\\
        &\leq\|P_i\|_2\|J_i^k\|_2\|P_i^{-1}\|_2\|x^{[0]}\|_2+\sum_{l=1}^k \|P_i\|_2\|J_i^{k-l}\|_2\|P_i^{-1}\|_2\|B_ir^{[l]}\|_2.
        \end{align*}
    Notice also that since $A$ is an upper triangular matrix with eigenvalues having absolute value at most $1$, the same applies to $A_i$ as it is obtained by removing rows and columns corresponding to variables not in the recurrence of $x_i$.
    Due to the restriction on the eigenvalues, we can rewrite $\|J_i^n\|_2= \|(I+(J_i-I))^n\|_2\leq \|(I+a_iN)^n\|_2$, where $a_i=\max(1,\sup_{j,o}|(J_i)_{j,o}|)$, and $N$ is an indicator matrix, where $N_{j,o}$ is $1$, when $(J_i)_{j,o}$ belongs to any Jordan block (excluding the diagonal) and $0$ otherwise. The inequality follows from $(I+a_iN)$ being element-wise larger in absolute value than $J_i$.

    Computing now the $n$-th power, we can look at the Jordan blocks separately, and for a Jordan block $(J_i)_m$ (with $N_m$ being the strictly triangular upper right one-matrix) with size $m_i$ (in the following $m$ for brevity) we obtain:
    \begin{align*}
        \|(J_i)_m^n\|_2\leq\|(I+a_iN_m)^n\|_2=\left\|\begin{pmatrix}
            1 & {n\choose{1}}a_i & {n\choose 2}a_i^2 &\dots & {n\choose{m-1}}a_i^{m-1}\\ 
            0& 1 & {n\choose{1}}a_i & \dots & {n\choose{m-2}}a_i^{m-2}\\ 
            \vdots&\vdots &\vdots & \ddots &\vdots\\
            0&0&0&1&{n\choose 1} a_i\\
            0&0&0&0& 1          
        \end{pmatrix}\right\|_2
    \end{align*}
    The Frobenius-norm~\cite{Horn_Johnson_1985_frobenius} forms an upper bound for the Euclidean norm, in particular for every matrix $\|M\|_2^2\leq \|M\|_F^2 = \sum_{i,j}M_{i,j}^2$. Therefore:
    \begin{align*}
        \|(J_i)_m^n\|_2^2&\leq \sum_{j=0}^{m-1}(m-j) \left({n\choose j}a_i^{j}\right)^2 
        \leq a_i^{2(m-1)}m \sum_{j=0}^{m-1} {n\choose j}^2\\
        &\leq a_i^{2(m-1)}m\left( \sum_{j=0}^{m-1} {\max(n,2(m-1))\choose m-1}^2 \right) 
        = a_i^{2(m-1)}m^2 {\max(n,2(m-1))\choose m-1}^2.
    \end{align*}
    Using the fact that ${n\choose{m}}\leq n^{m}$, we can conclude that 
    \begin{align*}
        \|(J_i)_m^n\|_2^2&\leq a_i^{2(m-1)}m^2 {\max(n,2(m-1))\choose m-1}^2 
        &\leq a_i^{2(m-1)}m^2 \left({n^{m-1}}+(2(m-1))^{m-1}\right)^2.
    \end{align*}
    Since the Euclidean norm of a matrix $M$ is defined as the square-root of the largest eigenvalue of $M^TM$, and $J_i^TJ_i$ is block-diagonal, its largest eigenvalue is the maximum of the eigenvalues of the blocks~\cite{george2024course_eigenvalues_block_diagonal}. Since $m$ is the dimension of the largest block, we conclude that:
    \[\|(J_i)^n\|_2\leq a_i^{(m-1)}m \left({n^{m-1}}+(2(m-1))^{m-1}\right).\]
    And with  $\delta_i := a^{(m-1)}_im \|P_i\|_2\|P_i^{-1}\|_2$, $\beta_i=\|B_i\|_2$ and $\gamma_i:=(2(m-1))^{m-1}$:
    \begin{align*}
        |x_i^{[k]}|&\leq\|P_i\|_2\|J_i^k\|_2\|P_i^{-1}\|_2\|x^{[0]}\|_2+\sum_{l=1}^k \|P_i\|_2\|J_i^{k-l}\|_2\|P_i^{-1}\|_2\|B_ir^{[l]}\|_2\\
        &\leq \delta_i(k^{m-1}+\gamma_i)\|x^{[0]}\|_2+\sum_{l=1}^k \delta_i((k-l)^{m-1}+\gamma_i)\beta_i\|r^{[l]}\|_2
    \end{align*}
    Applying this now to $\min(k,T)$, we get that,
    \begin{align*}
        |x_i^{[\min(k,T)]}|&\leq \delta_i(T^{m-1}+\gamma_i)\|x^{[0]}\|_2+\sum_{l=1}^T \delta_i((T-l)^{m-1}+\gamma_i)\beta_i\|r^{[l]}\|_2\\
        &\leq \delta_i(T^{m-1}+\gamma_i)\|x^{[0]}\|_2+\delta_i(T^{m-1}+\gamma_i)\beta_i\sum_{l=1}^T \|r^{[l]}\|_2.
    \end{align*}
    \end{proof}

    \theoremuniformintegrabilitylts*
    
    \begin{proof}
        We will first show that $(x^\alpha)^{[\min(k,T)]}$ is uniformly integrable.
            We do so by leveraging Theorem~\ref{thm:dominating_function_for_alpha_1}. W.l.o.g.~we assume $T\geq 1$:
            \begin{align*}&(x^\alpha)^{[\min(k,T)]}=\prod_{\alpha_i\in\alpha} \left(x_i^{[\min(k,T)]}\right)^{\alpha_i}\leq \prod_{\alpha_i\in\alpha}Z_i^{\alpha_i}\\
            &\leq \prod_{\alpha_i\in\alpha}\left(\delta_i(T^{m_i-1}+\gamma_i)\|x^{[0]}\|_2+\delta_i(T^{m_i-1}+\gamma_i)\beta_i\sum_{l=1}^T \|r^{[l]}\|_2\right)^{\alpha_i}\\
            &\leq \left( \prod_{\alpha_i\in\alpha}\delta_i^{\alpha_i}\right) \prod_{\alpha_i\in\alpha}\left((T^{m_i-1}+\gamma_i)\|x^{[0]}\|_2+(T^{m_i-1}+\gamma_i)\beta_i\sum_{l=1}^T \|r^{[l]}\|_2\right)^{\alpha_i}\\
            &\leq \left( \prod_{\alpha_i\in\alpha}\delta_i^{\alpha_i}\left(T^{m_i-1}\right)^{\alpha_i}\right)\prod_{\alpha_i\in\alpha}\left((1+\gamma_i)\|x^{[0]}\|_2+(1+\gamma_i)\beta_i\sum_{l=1}^T \|r^{[l]}\|_2\right)^{\alpha_i}\\
            &\leq \left( \prod_{\alpha_i\in\alpha}\delta_i^{\alpha_i}\left(T^{m_i-1}\right)^{\alpha_i}\right)\prod_{\alpha_i\in\alpha}2^{\alpha_i}\left(\left((1+\gamma_i)\|x^{[0]}\|_2\right)^{\alpha_i}+\left((1+\gamma_i)\beta_i\sum_{l=1}^T \|r^{[l]}\|_2\right)^{\alpha_i}\right)\\  
            &\leq \left( \prod_{\alpha_i\in\alpha}\delta_i^{\alpha_i}\left(T^{m_i-1}\right)^{\alpha_i}\right)\prod_{\alpha_i\in\alpha}(2+2\gamma_i)^{\alpha_i}\left(\left(\|x^{[0]}\|_2\right)^{\alpha_i}+\left(\beta_i\sum_{l=1}^T \|r^{[l]}\|_2\right)^{\alpha_i}\right)\\ 
            &\leq \left( \prod_{\alpha_i\in\alpha}\delta_i^{\alpha_i}(2+2\gamma_i)^{\alpha_i}\right)\prod_{\alpha_i\in\alpha}\left(\left(T^{m_i-1}\|x^{[0]}\|_2\right)^{\alpha_i}+\left(\beta_iT^{m_i-1}\sum_{l=1}^T \|r^{[l]}\|_2\right)^{\alpha_i}\right)\\
            &\leq \left( \prod_{\alpha_i\in\alpha}\delta_i^{\alpha_i}(2+2\gamma_i)^{\alpha_i}\max(1,\|x^{[0]}\|_2^{\alpha_i})\right)\prod_{\alpha_i\in\alpha}\left(\left(T^{m_i-1}\right)^{\alpha_i}+\left(\beta_iT^{m_i-1}\sum_{l=1}^T \|r^{[l]}\|_2\right)^{\alpha_i}\right)\\    
            \end{align*}
            Focusing now on the right product (the left one is just a number), and let $\hat{\beta}=\max_j\beta_j$. Then we can upper bound every term of the resulting expression:
    \begin{multline*}
        \prod_{\alpha_i\in\alpha}\left( \left(T^{m_i-1}\right)^{\alpha_i} +\left(\beta_iT^{m_i-1}\sum_{l=1}^T \|r^{[l]}\|_2\right)^{\alpha_i}\right)\\\leq 2^{|\alpha|}\left(T^{\sum_{\alpha_i\in\alpha}\alpha_i(m_i-1)}+T^{\sum_{\alpha_i\in\alpha}\alpha_i(m_i-1)}\left(\hat{\beta}\sum_{l=1}^T \|r^{[l]}\|_2\right)^{\sum_{\alpha_i\in\alpha}\alpha_i\mathbf{1}_{\{\beta_i>0\}}}\right).
    \end{multline*}
    Let now $N_\alpha = \sum\alpha_i(m_i-1)$ and $N_\beta=\sum_{\alpha_i\in\alpha}\alpha_i\mathbf{1}_{\{\beta_i>0\}}$. Notice that $N=N_\alpha+N_\beta$.

    Then by Theorem~\ref{thm: ost}, the sequence $\{(x^\alpha)^{[\min(k,T)]}\}$ is uniformly integrable, when $\mathbb{E}(T^{N_\alpha})$ is finite, and when $\mathbb{E}\left(T^{N_\alpha} \left(\sum_{l=1}^T \|r^{[l]}\|_2\right)^{N_\beta}\right)$ is finite. 
    This holds in fact by 
    Lemma~\ref{lemmaboundrandomsumindependent} or 
    Lemma~\ref{lemmaboundrandomsumdependent} 
    and our assumptions. 
    \end{proof}

    \begin{restatable}{lemma}{lemmaboundrandomsumindependent} 
    \label{lemmaboundrandomsumindependent}
        Let $N, M \in \mathbb{N}$. 
        Assume $\mathbb{E} \left( T^{N+M} \right) < \infty$ and that the random samples $r_j$ are stochastically independent of the runtime $T$ and $\mathbb{E} \left( \left|r_j \right|^{N} \right) < \infty$ for all $j \in \{1, 2, \cdots, m\}$. 
        Then 
        \[
            \mathbb{E} \left( T^M \left( \sum_{l = 1}^T \| r^{[l]} \|_2 \right)^N \right) < \infty.
        \]
    \end{restatable}
    \begin{proof}
        First, by the H\"older inequality on vectors, see e.g.~\cite[Theorem 10.23]{AnalyticInequalities}, for all $k \in \mathbb{N}$,
        \begin{align}
        \label{eq: Hoelder trick}
        \begin{split}
            \sum_{l = 1}^k \|r^{[l]} \|_2 
            &\leq 
            \left( \sum_{l = 1}^k \|r^{[l]} \|_2^N \right)^{1/N} \left( \sum_{l = 1}^k 1 \right)^{(N - 1)/N} 
            %\\
            %&
            = 
            k^{(N - 1)/N} 
            \left( \sum_{l = 1}^k \|r^{[l]} \|_2^N \right)^{1/N}.
            \end{split}
        \end{align}
        This proves that for all $k \in \mathbb{N}$, 
        \begin{align*}
            \left( \sum_{l = 1}^k \|r^{[l]} \|_2 \right)^N 
            \leq k^{N - 1} \sum_{l = 1}^k \|r^{[l]} \|_2^N. 
        \end{align*}
        In particular, this implies 
        \begin{align*}
            &\mathbb{E} \left(T^M \left( \sum_{l = 1}^T \|r^{[l]} \|_2 \right)^N \right) 
            = 
            \mathbb{E} \left( \sum_{k = 1}^{\infty} 1_{\{T = k\}} \cdot k^M \left( \sum_{l = 1}^k \|r^{[l]} \|_2 \right)^N \right) 
            \\
            \leq \ 
            &\mathbb{E} \left( \sum_{k = 1}^{\infty} 1_{\{T = k\}} \cdot k^{M + N - 1}  \sum_{l = 1}^k \|r^{[l]} \|_2^N \right) 
            = 
            \sum_{k = 1}^{\infty} k^{M + N - 1} \sum_{l = 1}^k \mathbb{E} \left( 1_{\{T = k\}} \cdot \|r^{[l]} \|_2^N \right).
        \end{align*}
        By assumption the random variables $1_{\{T = k\}}$ and $\|r^{[l]} \|_2^N$ are stochastically independent for all $k, l \in \mathbb{N}$. 
        Hence, by e.g.~\cite[Theorem 5.4]{klenke}, for all $k, l \in \mathbb{N}$, 
        \begin{align*}
            \mathbb{E} \left( 1_{\{T = k\}} \cdot \|r^{[l]} \|_2^N \right)
            =
            \mathbb{E} \left(1_{\{T = k\}} \right) \cdot
            \mathbb{E} \left( \|r^{[l]} \|_2^N \right)
            =
            \mathbb{P} \left( T = k \right) \cdot 
            \mathbb{E} \left( \|r^{[1]} \|_2^N \right),
        \end{align*}
        where for the last equality we have used that the random samples $\left( r^{[l]} \right)_{l \in \mathbb{N}}$ are i.i.d.. Therefore, we have proven that 
        \begin{align*}
            &\mathbb{E} \left(T^M \left( \sum_{l = 1}^T \|r^{[l]} \|_2 \right)^N \right) 
            \leq 
            \mathbb{E} \left( \|r^{[1]} \|_2^N \right)
            \sum_{k = 1}^{\infty} k^{M + N - 1} \sum_{l = 1}^k \mathbb{P} \left( T = k \right)
            \\
            &= 
            \mathbb{E} \left( \|r^{[1]} \|_2^N \right)
            \sum_{k = 1}^{\infty} k^{M + N} \cdot \mathbb{P} \left( T = k \right) 
            = 
            \mathbb{E} \left( \|r^{[1]} \|_2^N \right) \cdot \mathbb{E} \left( T^{N+M} \right).
        \end{align*}
        Now we have $\mathbb{E} \left( T^{N+M} \right) < \infty$ and 
        \begin{align*}
             \mathbb{E} \left( \|r^{[1]} \|_2^N \right) 
             &=\mathbb{E} \left( \left( \sum_{j = 1}^m |r_j^{[1]}|^2 \right)^{N/2} \right) 
             \leq 
             \mathbb{E} \left( \left( \sum_{j = 1}^m |r_j^{[1]}| \right)^{N} \right) 
             %\\ 
             %&
             \leq m^{N - 1} \sum_{j = 1}^m \mathbb{E} \left( |r_j^{[1]}|^N \right) < \infty,  
        \end{align*}
        where we have used a similar argument as \eqref{eq: Hoelder trick} and the assumption. 

        This concludes the proof. 
    \end{proof}

    \begin{restatable}{lemma}{lemmaboundrandomsumdependent}
    \label{lemmaboundrandomsumdependent}
                Let $N, M \in \mathbb{N}$. 
                Suppose $\mathbb{E} \left( T^{N+M+1} \right) < \infty$ and that the random samples $r_j$ satisfy for all $j \in \{1, 2, \cdots, m\}$, $\mathbb{E} \left( \left|r_j \right|^{\alpha} \right) < \infty$  
                where 
                \[
                    \alpha = \left\lceil \frac{N(N + M + 1)^2 + N}{M + N} \right\rceil.
                \]
        Then 
        \[
            \mathbb{E} \left(T^M \left( \sum_{l = 1}^T \|r^{[l]} \|_2 \right)^N \right) < \infty.
        \]
    \end{restatable} 

    \begin{proof}
        First, by the H\"older inequality on vectors, see e.g.~\cite[Theorem 10.23]{AnalyticInequalities}, for all $k \in \mathbb{N}$,
        \begin{align}
        \label{eq: Hoelder trick 2}
            \begin{split}
            \sum_{l = 1}^k \|r^{[l]} \|_2 
            &\leq 
            \left( \sum_{l = 1}^k \|r^{[l]} \|_2^N \right)^{1/N} \left( \sum_{l = 1}^k 1 \right)^{(N - 1)/N} 
            %\\ 
            %&
            = 
            k^{(N - 1)/N} 
            \left( \sum_{l = 1}^k \|r^{[l]} \|_2^N \right)^{1/N}.
            \end{split}
        \end{align}
        This proves that for all $k \in \mathbb{N}$, 
        \begin{align*}
            \left( \sum_{l = 1}^k \|r^{[l]} \|_2 \right)^N 
            \leq k^{N - 1} \sum_{l = 1}^k \|r^{[l]} \|_2^N. 
        \end{align*}
        In particular, this implies 
        \begin{align*}
            &\mathbb{E} \left( T^M \left( \sum_{l = 1}^T \|r^{[l]} \|_2 \right)^N \right) 
            = 
            \mathbb{E} \left( \sum_{k = 1}^{\infty} 1_{\{T = k\}} \cdot k^M \left( \sum_{l = 1}^k \|r^{[l]} \|_2 \right)^N \right) 
            \\
            \leq \ 
            &\mathbb{E} \left( \sum_{k = 1}^{\infty} 1_{\{T = k\}} \cdot k^{N + M - 1}  \sum_{l = 1}^k \|r^{[l]} \|_2^N \right) 
            = 
            \sum_{k = 1}^{\infty} k^{N + M - 1} \sum_{l = 1}^k \mathbb{E} \left( 1_{\{T = k\}} \cdot \|r^{[l]} \|_2^N \right).
        \end{align*}
        Now, choose 
        \begin{align}
            \label{eq: p}
            p &= \frac{(N + M + 1)^2 + 1}{(N + M)^2 + N + M + 2} \in (1, \infty),
            \\
            q &= \frac{(N + M + 1)^2 + 1}{N + M} \in (1, \infty).
            \label{eq: q}
        \end{align}
        In particular, $\frac{1}{p} + \frac{1}{q} = 1$. 
        Then by H\"older's Inequality for Random Variables, see for instance \cite[Theorem 7.16]{klenke}, we have for all $k, l \in \mathbb{N}$, 
        \begin{align*}
            \mathbb{E} \left( 1_{\{T = k\}} \cdot \|r^{[l]} \|_2^N \right) 
            &\leq 
            \mathbb{E} \left( 1_{\{T = k\}}^p \right)^{1/p} 
            \cdot \mathbb{E} \left( \|r^{[l]} \|_2^{q N} \right)^{1/ q} 
            = 
            \mathbb{P} \left( T = k \right)^{1/p} 
            \cdot \mathbb{E} \left( \|r^{[1]} \|_2^{q N} \right)^{1/ q}, 
        \end{align*}
        where we have used that the random samples $\left( r^{[l]} \right)_{l \in \mathbb{N}}$ are i.i.d..
        Combining both estimates, 
        \begin{align*}
            \mathbb{E} \left( T^M \left( \sum_{l = 1}^T \|r^{[l]} \|_2 \right)^N \right) 
            &\leq 
            \sum_{k = 1}^{\infty} k^{N + M - 1} \sum_{l = 1}^k \mathbb{E} \left( 1_{\{T = k\}} \cdot \|r^{[l]} \|_2^N \right) 
            \\
            &\leq 
            \mathbb{E} \left( \|r^{[1]} \|_2^{q N} \right)^{1/ q}
            \sum_{k = 1}^{\infty} k^{N + M} \cdot \mathbb{P} \left( T = k \right)^{1/p}. 
        \end{align*}
        Regarding the random samples we have by an argument similar to \eqref{eq: Hoelder trick 2},
        \begin{align*}
             \mathbb{E} \left( \|r^{[1]} \|_2^{qN} \right) 
             &=  \mathbb{E} \left( \left( \sum_{j = 1}^m |r_j^{[1]}|^2 \right)^{qN/2} \right) 
             \leq 
             \mathbb{E} \left( \left( \sum_{j = 1}^m |r_j^{[1]}| \right)^{qN} \right) 
             %\\ 
             %&
             \leq 
             m^{qN - 1} \sum_{j = 1}^m \mathbb{E} \left( |r_j^{[1]}|^{qN} \right) 
             \\ 
             &= 
             m^{(N(N + M + 1)^2 - M) / (N + M)} \sum_{j = 1}^m \mathbb{E} \left( |r_j^{[1]}|^{(N(N + M + 1)^2 + N) / (N + M)} \right) 
             < \infty,  
        \end{align*}
        where have used \eqref{eq: q} and that the last inequality holds by assumption.

        Therefore, it remains to show that 
        \begin{align}
        \label{eq: series bound}
            \sum_{k = 1}^{\infty} k^{N + M} \cdot \mathbb{P} \left( T = k \right)^{1/p} < \infty.
        \end{align}
        Now choose $a := (N + M + 1)/(N + M) \in (1, \infty)$ and $b := N + M + 1$, i.e.~$\frac{1}{a} + \frac{1}{b} = 1$.
        Then by H\"older's Inequality for power series, see e.g.~\cite[Theorem 10.23]{AnalyticInequalities},
        \begin{align*}
            &\sum_{k = 1}^{\infty} k^{N + M}  \mathbb{P} \left( T = k \right)^{1/p} 
            = 
            \sum_{k = 1}^{\infty} \left( k^{N + M} \mathbb{P} \left( T = k \right)^{(N + M)/(N + M + 1)} \right) \mathbb{P} \left( T = k \right)^{1/p - (N + M)/(N + M + 1)}
            \\
            &\leq 
            \left( \sum_{k = 1}^{\infty} \left( k^{N+M} \mathbb{P} \left( T = k \right)^{(N+M)/(N + M + 1)} \right)^a \right)^{1/a} 
            \left( \sum_{k = 1}^{\infty} \mathbb{P} \left( T = k \right)^{b (1/p - (N+M)/(N + M + 1))} \right)^{1/b}.
        \end{align*}
        For the first term note that as $a = (N + M + 1)/(N + M)$ and as $\mathbb{E}(T^{N+M+1}) < \infty$,  
        \begin{align*}
            \sum_{k = 1}^{\infty} \left( k^{N+M} \cdot \mathbb{P} \left( T = k \right)^{(N + M)/(N + M + 1)} \right)^a 
            &= 
            \sum_{k = 1}^{\infty} k^{N + M +1} \cdot \mathbb{P} \left( T = k \right) 
            = \mathbb{E}(T^{N+M+1}) < \infty.
        \end{align*}
    %    Hence, this proves that
    %    \begin{align*}
    %        \sum_{k = 1}^{\infty} k^{N + M} \cdot \mathbb{P} \left( T = k \right)^{1/p}  
     %       \leq 
     %       \mathbb{E}(T^{N+M+1})^{(N+M)/(N+M+1)} 
     %       \left( \sum_{k = 1}^{\infty} \mathbb{P} \left( T = k \right)^{b (1/p - N/(N + 1))} \right)^{1/b}. 
     %   \end{align*}
     %
        For the second term note that as $b = N + M + 1$, 
        \begin{align}
        \label{eq: series bound 0}
            \sum_{k = 1}^{\infty} \mathbb{P} \left( T = k \right)^{b (1/p - (N + M)/(N + M + 1))} 
            =
            \sum_{k = 1}^{\infty} \mathbb{P} \left( T = k \right)^{ (N + M + 1)/p - (N + M)}.
        \end{align}
        Moreover, by \eqref{eq: p} we have 
        \begin{align*}
            \frac{N + M + 1}{p} - N - M
            &=
            \frac{(N+M+1)((N + M)^2 + N + M + 2)}{(N+M +1)^2 + 1} - N - M
            \\
            &=
            \frac{N + M + 2}{(N + M + 1)^2 + 1} > 0.
        \end{align*}
        Furthermore, by Markov's Inequality, see for instance \cite[Theorem 5.11]{klenke}, we have for all $k \in \mathbb{N}$, 
        \begin{align*}
            \mathbb{P} \left( T = k \right) 
            \leq 
            \mathbb{P} \left( T \geq k \right) 
            \leq 
            \frac{\mathbb{E}(T^{N+M+1})}{k^{N + M + 1}}.
        \end{align*}
        Combining these two equations we obtain 
        \begin{align*}
            \sum_{k = 1}^{\infty} \mathbb{P} \left( T = k \right)^{ (N + M + 1)/p - N - M}
            \leq 
            \mathbb{E} \left( T^{N+M+1} \right)^{ \frac{N + M + 2}{(N + M +1)^2 + 1}} 
            \sum_{k = 1}^{\infty}
            k^{-\frac{(N + M + 1)(N + M + 2)}{(N + M + 1)^2 + 1}}.
         \end{align*}
         And, 
         \begin{align*}
        \frac{(N + M + 1)(N + M + 2)}{(N + M + 1)^2 + 1} 
         &= \frac{(N + M)^2 + 3N + 3M + 2}{(N + M)^2 + 2N + 2M + 2} 
         \\
         &= 1 + \frac{N + M}{(N + M)^2 + 2N + 2M + 2} > 1.
         \end{align*}
         This proves indeed that 
        \begin{align*}
            &\sum_{k = 1}^{\infty} \mathbb{P} \left( T = k \right)^{ (N + M + 1)/p - N - M}
            %\\
            %&
            \leq 
            \mathbb{E} \left(T^{N+M+1}\right)^{ \frac{N + M + 2}{(N + M +1)^2 + 1}} 
            \sum_{k = 1}^{\infty}
            k^{-\frac{(N + M + 1)(N + M + 2)}{(N + M + 1)^2 + 1}} < \infty. 
         \end{align*}
        Hence, by \eqref{eq: series bound 0} we have shown that 
        \begin{align*}
            \sum_{k = 1}^{\infty} \mathbb{P} \left( T = k \right)^{b (1/p - (N + M)/(N + M + 1))} < \infty.
        \end{align*}
         This is the second inequality we had to prove. 
         This concludes the proof of \eqref{eq: series bound}. 
         In particular, this also proves the claim overall. 
    \end{proof}

    \clearpage
    \section{Supplementary Material for Bound Derivation}
    \label{appendix:derivation-rules}
    In this section we prove that the derivation rules are probability-theoretically sound.
        \begin{mathpar}
              \inferrule*[right=(jensen-lb-1)]
        {a\leq \mathbb{E}(\vec{x}^\alpha), 0 \leq a}
        {a^2\leq \mathbb{E}(\vec{x}^{2\alpha})}     
        \and
        \inferrule*[right=(jensen-lb-2)]
        {\mathbb{E}(\vec{x}^\alpha)\leq a, a\leq 0}
        {a^2\leq \mathbb{E}(\vec{x}^{2\alpha})}
        \end{mathpar}

        \begin{lemma}
            The rules (jensen-lb-1) and (jensen-lb-2) are sound.
        \end{lemma}
        \begin{proof}
            This follows from Jensen's inequality~\cite{klenke}, and the fact that the square function $f(X)=X^{2}$ is convex. 

            We therefore have $a^2\leq\mathbb{E}(X^2)$ whenever $a$ is a positive lower bound for $X$, or when $a$ is a negative upper bound.
        \end{proof}

        \begin{mathpar}
            \inferrule*[right=(rv-mul-1)]
        {a\leq \vec{x}_T^\alpha \\ a\leq 0 \\ 0\leq \vec{x}_T^\beta\\ \mathbb{E}(\vec{x}_T^\beta) \leq b}
        {ab\leq \mathbb{E}(\vec{x}_T^{\alpha+\beta})}
        \end{mathpar}
        \begin{lemma}
            The rule (rv-mul-1) is sound.
        \end{lemma}
        \begin{proof}
            Let us denote by $X$ the random variable $\vec{x}_T^\alpha$, and with $Y$ the random variable $\vec{x}_T^\beta$. By the premise we have that $\forall\omega\in\Omega:a\leq X(\omega)$, and $\forall\omega\in\Omega:0\leq Y(\omega)$.
        
            It holds then $\forall\omega\in\Omega:aY(\omega)\leq X(\omega)Y(\omega)$ (notice that positivity of $Y(\omega)$ and negativity of $a$ are used here for this inequality to hold).
            Taking now the expected value on both sides, we get then that $ab\leq a\mathbb{E}(Y)\leq \mathbb{E}(XY)$.
        \end{proof}

        \begin{mathpar}
        \inferrule*[right=(rv-mul-2)]
        {\vec{x}_T^\alpha\leq a \\ a\geq 0 \\ 0\leq \vec{x}_T^\beta\\ \mathbb{E}(\vec{x}_T^\beta) \leq b\\ 0\leq b}
        {\mathbb{E}(\vec{x}_T^{\alpha+\beta})\leq ab}
        \end{mathpar}
        \begin{lemma}
            The rule (rv-mul-2) is sound.
        \end{lemma}
        \begin{proof}
            Let us denote by $X$ the random variable $\vec{x}_T^\alpha$, and with $Y$ the random variable $\vec{x}_T^\beta$. By the premise we have that $\forall\omega\in\Omega:X(\omega) \leq a$ with $a\geq 0$ and $\forall\omega\in\Omega:0\leq Y(\omega)$.

            It then holds that $\forall\omega\in\Omega: X(\omega)Y(\omega)\leq aY(\omega)$ (we use positivity of $Y(\omega)$ and $a$ in this inequality). Taking then the expected value, we get $\mathbb{E}(XY)\leq a\mathbb{E}(Y)\leq ab$.
        \end{proof}

        \begin{mathpar}
            \inferrule*[right=(rv-mul-3)]
        {a\leq \vec{x}_T^\alpha \\ a\geq 0 \\ 0\leq \vec{x}_T^\beta\\ b\leq\mathbb{E}(\vec{x}_T^\beta)\\ 0\leq b}
        {ab\leq\mathbb{E}(\vec{x}_T^{\alpha+\beta})}
        \end{mathpar}
        \begin{lemma}
            The rule (rv-mul-3) is sound.
        \end{lemma}
        \begin{proof}
            Let us again denote by $X$ the random variable $\vec{x}_T^\alpha$, and with $Y$ the random variable $\vec{x}_T^\beta$. By the premise we have that $\forall\omega\in\Omega:0\leq a\leq X(\omega)$ and $\forall\omega\in\Omega:0\leq Y(\omega)$.

            It then holds that $\forall\omega\in\Omega: aY(\omega)\leq X(\omega)Y(\omega)$ (using again positivity of $Y(\omega)$ and $a$).
            Therefore $ab\leq a\mathbb{E}(Y)\leq\mathbb{E}(XY)$.
        \end{proof}

    \begin{mathpar}
        \inferrule*[right=(cs-lb)]
        {\mathbb{E}((a\vec{x}^\alpha+b\vec{x}^\beta)^2) \leq g\\ \mathbb{E}(\vec{x}^{2\alpha})\leq h}
        {-\frac{\sqrt{gh}}{|b|}-\frac{h}{2}\left(\left|\frac{a}{b}\right|+\frac{a}{b}\right)\leq \mathbb{E}(\vec{x}^{\alpha+\beta})}
    \end{mathpar}
    \begin{lemma}
    \label{lemma:rule-cs-lb-sound}
        The rule (cs-lb) is sound.
    \end{lemma}
    \begin{proof}
        Let $X$ and $Y$ be the random variables $\vec{x}^\alpha$ and $\vec{x}^\beta$ respectively. It follows that \[\mathbb{E}(XY)=\mathbb{E}\left(X\frac{aX+bY-aX}{b}\right)=\frac{1}{b}\mathbb{E}(X(aX+bY))-\frac{a}{b}\mathbb{E}(X^2)\]
        By the Cauchy-Schwarz inequality, $|\mathbb{E}(X(aX+bY))|\leq\sqrt{\mathbb{E}(X^2)\mathbb{E}((aX+bY)^2)}\leq\sqrt{gh}$, thus $\frac{1}{b}\mathbb{E}(X(aX+bY))\geq -\frac{1}{|b|}\sqrt{gh}$.
        We now proceed with the second term: Since $\mathbb{E}(X^2)\geq0$, $-\frac{a}{b}\mathbb{E}(X^2)\geq -\max(0,\frac{a}{b}\mathbb{E}(X^2))=-\frac{1}{2}(\left|\frac{a}{b}\right|+\frac{a}{b})\mathbb{E}(X^2)\geq- \frac{h}{2}(\left|\frac{a}{b}\right|+\frac{a}{b})$.

        Putting this together, we get $\mathbb{E}(\vec{x}^{\alpha+\beta})=\mathbb{E}(XY)\geq -\frac{\sqrt{gh}}{|b|}-\frac{h}{2}(\left|\frac{a}{b}\right|+\frac{a}{b})$.
    \end{proof}

    \begin{mathpar}
        \inferrule*[right=(cs-ub)]
        {\mathbb{E}((a\vec{x}^\alpha+b\vec{x}^\beta)^2) \leq g\\ \mathbb{E}(\vec{x}^{2\alpha})\leq h}
        {\mathbb{E}(\vec{x}^{\alpha+\beta})\leq \frac{\sqrt{gh}}{|b|}+\frac{h}{2}\left(\left|\frac{a}{b}\right|-\frac{a}{b}\right)}
    \end{mathpar}
    \begin{lemma}
        The rule (cs-ub) is sound.
    \end{lemma}
    \begin{proof}
        Similar to the proof of Lemma~\ref{lemma:rule-cs-lb-sound}, we start with
        \[\mathbb{E}(XY)=\frac{1}{b}\mathbb{E}(X(aX+bY))-\frac{a}{b}\mathbb{E}(X^2)\]
        and by Cauchy-Schwarz $|\mathbb{E}(X(aX+bY))|\leq\sqrt{gh}$, thus $\frac{1}{b}\mathbb{E}(X(aX+bY))\leq \frac{1}{|b|}\sqrt{gh}$.
        Again since $\mathbb{E}(X^2)\geq 0$, we get that $-\frac{a}{b}\mathbb{E}(X^2)\leq \max(0,-\frac{a}{b}\mathbb{E}(X^2))=\frac{1}{2}(|\frac{a}{b}|-\frac{a}{b})\mathbb{E}(X^2)\leq\frac{h}{2}(|\frac{a}{b}|-\frac{a}{b})$.
        Again, putting this together we get $\mathbb{E}(\vec{x}^{\alpha+\beta})=\mathbb{E}(XY)\leq \frac{\sqrt{gh}}{|b|}+\frac{h}{2}\left(\left|\frac{a}{b}\right|-\frac{a}{b}\right)$
    \end{proof}

    \begin{mathpar}
        \inferrule*[right=(mk-ub)]
        {\mathbb{E}((a\vec{x}^\frac{\alpha}{2}+b\vec{x}^\frac{\beta}{2})^2) \leq u\\ \mathbb{E}(\vec{x}^\beta)\leq v}
    {\mathbb{E}(\vec{x}^\alpha) \leq \frac{1}{a^2}u + 2 \frac{|b|}{a^2}\sqrt{u}\sqrt{v} + \frac{b^2}{a^2}v}
    \end{mathpar}
    \begin{lemma}
        The rule (mk-ub) is sound.
    \end{lemma}
    \begin{proof}
        Let us denote $X:=x^\frac{\alpha}{2}$ and $Y:=x^\frac{\beta}{2}$.
        Then by the definition of the $L^2$ norm we have $\|aX+bY\|_2\leq \sqrt{u}$ and furthermore $\|Y\|_2\leq\sqrt{v}$ (notice that $u$ and $v$ must be positive, as they are upper bounds of squares, since $\beta$ is implicitly an even multi-index, as $x^{\frac{\beta}{2}}$ is a monomial).

        It holds in general that $aX = (aX+bY) -bY$. By the Minkowski inequality~\cite{klenke}, it then holds that $\|aX\|_2\leq\|aX+bY\|_2+\|-bY\|_2 =\sqrt{u}+|b|\sqrt{v}$. Hence, $\|X\|_2\leq \frac{\sqrt{u}+|b|\sqrt{v}}{|a|}$.

        Since $\mathbb{E}(x^\alpha) = \mathbb{E}(X^2)=\|X\|_2^2\leq (\frac{\sqrt{u}+|b|\sqrt{v}}{|a|})^2=\frac{u+2|b|\sqrt{u}\sqrt{v}+b^2 v}{a^2}$, we have proven the claim.
    \end{proof}

    \subsection{Limitations of Bound Derivation}
    \label{appendix:limitations-of-implementation}
    The current implementation scales very badly when multiple bounds exist for the same monomial. This is the result of the martingale-rules being applied for all different combinations of bounds. To deal with that, it is important to have effective checks that prevent unnecessary bounds from being added/kept. This is an even bigger challenge, when adding additional rules that introduce terms that are not polynomials.

    Another challenge is the synthesis of martingale expressions. The equation system defined in Eq.~\eqref{eq:linear-system-of-equations} is usually under-determined, and not all solutions are equally useful for deriving bounds. One option would be to enumerate all solutions, with a unique set of zero-coefficients, but this enumeration is costly, and would also introduce a lot of rules. We therefore use mixed-integer linear programming and the \texttt{CBC}-solver~\cite{cbc_solver}, to iteratively fix one monomial to be nonzero, and compute all (or a user-provided maximum number of) solutions with a maximum number of zero-coefficients.

    Lastly, we note that the introduced multiplication rules for random variables require strict knowledge about the sign of the random variables involved. In our example, information about this sign can easily be obtained from the negated loop guard. When generalising our approach to a loop guard of form $x>b$, one would also have to consider expressions of the form $(x-b)$ and $\mathbb{E}((x-b))$, in order to make use of the multiplication rules. The presence of such expressions should also be considered in the martingale synthesis process.

    %\clearpage
    \section{Termination Proofs of Examples}
    \label{appendix:stopping-times-of-examples}
    In this section we show that the stopping times of the examples behave as required in Table~\ref{tab:experimental-results}.
    \paragraph{Example in Figure~\ref{fig:running-example-linear-loop}.} We define a stochastic process with $\{U_n\}_{n\in\mathbb{N}}$ being a sequence of i.i.d.~random variables drawn from $\text{Uniform}(-1,0)$. We then define $\{X_n\}_{n\in\mathbb{N}}=x_0+\sum_{i=1}^n U_i$.
    It is evident that $\{X_n\}$ corresponds to $x$ in the $n$-th iteration of Figure~\ref{fig:running-example-linear-loop} as long as the loop is still running. Therefore, $T:=\inf \{n\in\mathbb{N} \mid X_n<0\}$ captures its termination behaviour.

    We can now construct a martingale $\{M_n\}_{n\in\mathbb{N}}$, where $M_n=X_n+\frac{n}{2}$ and use Azuma's inequality~\cite{klenke} to derive the following bound: 
    \begin{align*}
        \mathbb{P}(X_n\geq 0) = \mathbb{P} \left( M_n\geq \frac{n}{2}-x_0 \right) \leq\begin{cases}
            \exp\left(\frac{-(\frac{n}{2}-x_0)^2}{2\sum_{k=1}^n c_k^2}\right) & \text{if $\frac{n}{2}-x_0\geq0$}\\
            1 &\text{else},
        \end{cases} 
    \end{align*}
    where $c_k$ is an almost sure difference bound of the martingale. In our case this bound is $0.5$.
    Hence, 
    \begin{align*}
        \mathbb{P}(X_n\geq 0) \leq
        \begin{cases}
           \exp\left(\frac{-\frac{n^2}{4}+x_0n-x_0^2}{n/2}\right) & \text{if $\frac{n}{2}-x_0\geq0$}\\
            1 &\text{else}
        \end{cases}
    \end{align*}
    This immediately gives us a bound on $\mathbb{P}(T\geq n)$, since $T \geq n$ implies that $X_n\geq 0$.
    Hence,
    \[
        \mathbb{E}\left(T^N\right)\leq \sum_{n=1}^\infty \mathbb{P}(T\geq \sqrt[N]{n}) \leq (2x_0)^N+\sum_{n=(2x_0)^N+1}^\infty \exp\left(\frac{-\frac{\sqrt[N]{n^2}}{4}+x_0\sqrt[N]{n}-x_0^2}{\sqrt[N]{n}/2}\right).
    \]
    Notice that the exponent is in $\exp(\mathcal{O}(-\sqrt[N]{n}))$, thus the sum converges for any fixed $N$, hence every moment of the stopping time is finite.

    \paragraph{Example in Figure~\ref{fig:example-unbounded-update-for-lg}.} 
    Define $\{Z_n\}_{n\in\mathbb{N}}$ to be a sequence of random variables from $\text{Normal}(0,1)$. Let $\{B_n\}_{n\in\mathbb{N}}$ be a sequence of variables, which are $-1$ with probability $\frac{7}{10}$, and $1$ else. Notice that then again the sequence $\{X_n\}_{n\in\mathbb{N}}$ with $X_n=x_0+\sum_{i=1}^n Z_i +\sum_{i=1}^n B_i$ describes the random variable $x$ in the $n$-th iteration before the loop has stopped. Therefore, $T:=\inf \{n\in\mathbb{N} \mid X_n<0\}$ captures the loop's termination behaviour.

    First we compute the moment generating function of $X_n - \mathbb{E}(X_n)$. 
    Note that by definition of the moment generating function and the independence of the variables $Z_i$, $X_i$, for all $t \geq 0$, 
    \begin{align}
    \label{eq: moment generating function}
        \begin{split}
        \mathcal{M}_{X_n-\mathbb{E}(X_n)}(t) 
        &= \mathbb{E} \left[ \exp \left( \left( X_n-\mathbb{E}(X_n) \right) t \right) \right] 
        \\
        &= \mathbb{E} \left[ \exp \left( \left( \sum_{i = 1}^n Z_i \right) t \right) \right] \cdot \prod_{i = 1}^n \mathbb{E} \left[ \exp \left( \left( B_i - \mathbb{E}(B_i) \right) t \right) \right].
        \end{split}
    \end{align}
    Now by using the additivity of Normal distributions we have $\sum_{i = 1}^n Z_i \sim \text{Normal}(0,n)$. 
    Therefore, by known properties of Gaussian random variables we have for all $t \geq 0$,
    \begin{align*}
        \mathbb{E} \left[ \exp \left( \left( \sum_{i = 1}^n Z_i \right) t \right) \right] = \exp \left( \frac{n t^2}{2} \right)
    \end{align*}
    Moreover, by Hoeffding's Lemma~\cite{massart2007concentration} applied to the random variable $B_i$ we obtain for all $t \geq 0$,
    \begin{align*}
        \mathbb{E} \left[ \exp \left( \left( B_i - \mathbb{E}(B_i) \right) t \right) \right] 
        = \mathcal{M}_{B_i-\mathbb{E}(B_i)}(t) 
        \leq \exp \left( \frac{4 t^2}{8} \right).
    \end{align*}
    Then by inserting this into \eqref{eq: moment generating function} we obtain that for all $t \geq 0$,
%
   % We then can find an upper bound of the moment generating function of $X_n$ using the additivity of Normal distributions, and Hoeffding's Lemma~\cite{massart2007concentration} for the sequence $\{B_n\}$.
%
    \begin{align*}
        \mathcal{M}_{X_n-\mathbb{E}(X_n)}(t) &\leq \exp\left(\frac{nt^2}{2}\right) \prod_{i=1}^n\exp\left(\frac{t^24}{8}\right) 
        =\exp\left(\frac{nt^2}{2}+\frac{4nt^2}{8}\right) = \exp\left(\frac{2nt^2}{2}\right)
    \end{align*}
    This means that $X_n-\mathbb{E}(X_n)$ is sub-Gaussian with variance proxy $2n$.
    
    We further have $\mathbb{E}(X_n)=n \left(-\frac{7}{10}+\frac{3}{10}\right)=-\frac{2}{5}n$. %and $\sigma^2_{B_n} = n+n(\frac{7}{10}+\frac{3}{10})=2n$.

    Hence, by the Chernoff Bound, we have that
    \begin{align*}
        \mathbb{P}(X_n\geq0) 
        &= \mathbb{P}\left(X_n-x_0+\frac{2n}{5} \geq \frac{2n}{5}-x_0\right) 
        = \mathbb{P}\left(X_n - \mathbb{E}(X_n) \geq \frac{2n}{5}-x_0 \right) 
        \\
        &\leq \inf_{t \geq 0} \left( \mathcal{M}_{X_n-\mathbb{E}(X_n)}(t) \cdot \exp \left( -t  \left( \frac{2n}{5}-x_0 \right) \right) \right)
        \\ 
        &\leq \inf_{t \geq 0} \exp \left( n t^2 -t  \left( \frac{2n}{5}-x_0 \right) \right) 
        \leq \exp \left( - \frac{(2n - 5 x_0)^2}{100n} \right).
        %\exp\left(\frac{-\frac{(2n-5x_0)^2}{25}}{4n}\right)
    \end{align*}
    %$\mathbb{P}(X_n\geq0)=\mathbb{P}(X_n-x_0+\frac{2n}{5}\geq \frac{2n}{5}-x_0)\leq \exp\left(\frac{-\frac{(2n-5x_0)^2}{25}}{4n}\right)$.
    Similarly to the previous bound, we get that:
$\mathbb{E}(T^N)\leq \sum_{n=1}^\infty \mathbb{P}(T\geq \sqrt[N]{n}) \leq \sum_{n=1}^\infty\exp(\mathcal{O}(-\sqrt[N]{n}))$, thus the sum converges for any fixed $N$, hence every moment of the stopping time is finite. 
    
\end{document}